%% file: main.tex
\documentclass[11pt]{article}

\input{header_arxiv}

\input{def}

\title{The Tamed Unadjusted Langevin Algorithm}

\author{Nicolas Brosse \textsuperscript{1}, Alain Durmus \textsuperscript{2}, \'Eric Moulines \textsuperscript{1} and Sotirios Sabanis \textsuperscript{3}}
\date{\today}

\begin{document}

\footnotetext[1]{Centre de Math\'ematiques Appliqu\'ees, UMR 7641, Ecole Polytechnique, France. \\
Emails: nicolas.brosse@polytechnique.edu, eric.moulines@polytechnique.edu }
\footnotetext[2]{Ecole Normale Sup\'erieure
CMLA
61, Av. du Pr\'esident Wilson
94235 Cachan Cedex, France \\
Email: alain.durmus@cmla.ens-cachan.fr}
\footnotetext[3]{
University of Edinburgh, Scotland, UK.
Email: s.sabanis@ed.ac.uk }

\maketitle

\begin{abstract}
In this article, we consider the problem of sampling from a probability measure $\pi$ having a density on $\mathbb{R}^d$ proportional to $x\mapsto \mathrm{e}^{-U(x)}$.
% / \int_{\mathbb{R}^d} \mathrm{e}^{-U(y)} \mathrm{d} y$.
The Euler discretization of the Langevin stochastic differential equation (SDE) is known to be unstable,
when the potential $U$ is superlinear.
%, \ie~$\liminf_{\Vert x \Vert \to +\infty} \Vert \nabla U(x)\Vert / \Vert x \Vert = +\infty$.
Based on previous works on the taming of superlinear drift coefficients for SDEs, we introduce the Tamed Unadjusted Langevin Algorithm (TULA) and obtain non-asymptotic bounds in $V$-total variation norm and Wasserstein distance of order $2$ between the iterates of TULA and $\pi$, as well as weak error bounds.
Numerical experiments are presented which support our findings.
\end{abstract}

%\linenumbers

\input{introduction}
\input{notations}
\input{part1}

\input{examples}
\input{proof}
%\input{parts/old-proof-strongly-convex-wasserstein}

%\appendix
%
%\input{parts/annexe}

\def\sectionautorefname{Section}
\def\subsectionautorefname{Section}
\def\subsubsectionautorefname{Section}
\def\corollaryautorefname{Corollary}

\section*{Acknowledgements}

This work was supported by the \'Ecole Polytechnique Data Science Initiative and the Alan Turing Institute under the EPSRC grant EP/N510129/1.	

\appendix

\input{appendix}

%\section*{References}
%
%\bibliography{../Bibliography/bibliography}

\printbibliography

\end{document}

%% file: header_arxiv.tex
\usepackage[a4paper]{geometry} % required for statsoc+pdflatex

\usepackage[utf8]{inputenc}
\usepackage{amssymb}
\usepackage{makeidx}
\usepackage[english]{babel}
\usepackage{graphicx}
\usepackage{amsfonts,amsmath,amssymb,amsthm}

\usepackage{mathrsfs}
\usepackage[active]{srcltx}
\usepackage{verbatim}
\usepackage{enumitem} %label option for enumerate \alph* \Alph* \roman* \Roman* or other... [label=$\bullet$], [label=\alph*]
\usepackage[backend=biber, style=alphabetic,doi=false,isbn=false,url=false]{biblatex}
\usepackage[toc,page]{appendix}
\usepackage{aliascnt,bbm}
\usepackage{array}
\usepackage[colorlinks = true, citecolor = blue]{hyperref}
\usepackage[textwidth=3cm,textsize=footnotesize]{todonotes}
\usepackage{xargs}
\usepackage{cellspace}
\usepackage{color}

\usepackage[Symbolsmallscale]{upgreek}
\usepackage{array,multirow,makecell}
\setcellgapes{4pt}
\makegapedcells
\newcolumntype{R}[1]{>{\raggedleft\arraybackslash }b{#1}}
\newcolumntype{L}[1]{>{\raggedright\arraybackslash }b{#1}}
\newcolumntype{C}[1]{>{\centering\arraybackslash }b{#1}}

\usepackage{accents}
\usepackage{dsfont}
\usepackage{aliascnt}
\usepackage{cleveref}
\makeatletter
\newtheorem{theorem}{Theorem}
\crefname{theorem}{theorem}{Theorems}
\Crefname{Theorem}{Theorem}{Theorems}

\newaliascnt{lemma}{theorem}
\newtheorem{lemma}[lemma]{Lemma}
\aliascntresetthe{lemma}
\crefname{lemma}{lemma}{lemmas}
\Crefname{Lemma}{Lemma}{Lemmas}

\newaliascnt{corollary}{theorem}

\aliascntresetthe{corollary}
\crefname{corollary}{corollary}{corollaries}
\Crefname{Corollary}{Corollary}{Corollaries}

\newaliascnt{proposition}{theorem}
\newtheorem{proposition}[proposition]{Proposition}
\aliascntresetthe{proposition}
\crefname{proposition}{proposition}{propositions}
\Crefname{Proposition}{Proposition}{Propositions}

\newaliascnt{definition}{theorem}

\aliascntresetthe{definition}
\crefname{definition}{definition}{definitions}
\Crefname{Definition}{Definition}{Definitions}

\newaliascnt{definition-proposition}{theorem}

\aliascntresetthe{definition-proposition}
\crefname{definition-proposition}{definition-proposition}{definitions-propositions}
\Crefname{Definition-Proposition}{Definition-Proposition}{Definitions-Propositions}

\newaliascnt{remark}{theorem}
\newtheorem{remark}[remark]{Remark}
\aliascntresetthe{remark}
\crefname{remark}{remark}{remarks}
\Crefname{Remark}{Remark}{Remarks}

\crefname{example}{example}{examples}
\Crefname{Example}{Example}{Examples}

\crefname{figure}{figure}{figures}
\Crefname{Figure}{Figure}{Figures}

\newtheorem{assumption}{\textbf{H}\hspace{-3pt}}
\Crefname{assumption}{\textbf{H}\hspace{-3pt}}{\textbf{H}\hspace{-3pt}}
\crefname{assumption}{\textbf{H}}{\textbf{H}}

\Crefname{probleme}{\textbf{Problem}\hspace{-3pt}}{\textbf{Problem}\hspace{-3pt}}
\crefname{probleme}{\textbf{Problem}}{\textbf{Problem}}

\newtheorem{assumptionA}{\textbf{A}\hspace{-3pt}}
\Crefname{assumptionA}{\textbf{A}\hspace{-3pt}}{\textbf{A}\hspace{-3pt}}
\crefname{assumptionA}{\textbf{A}}{\textbf{A}}

\Crefname{assumptionG}{\textbf{G}\hspace{-3pt}}{\textbf{G}\hspace{-3pt}}
\crefname{assumptionG}{\textbf{G}}{\textbf{G}}

\makeatother

%% file: def.tex
\def\Lspace{\operatorname{L}}

\newcommandx{\Uga}[1][1=]{\ifthenelse{\equal{#1}{}}{U^{\gaStep}}{ U^{#1}} }
\newcommand{\Ugad}{\dr}
\def\pU{U}
\def\nablaU{\nabla \pU}
\def\invpi{\pi}
\def\invpig{\pi_\gastep}
\def\LG{L}
\def\mU{m}
\def\kappaU{\kappa}

\def\LH{L_H}
\def\nue{\nu}
\def\betah{\beta}
\def\CH{C_{H}}

\newcommandx{\dr}[1][1=]{\ifthenelse{\equal{#1}{}}{G_{\gaStep}}{ G_{\gastep,#1}} }
\def\Cal{C_\alpha}
\def\alphag{\alpha}

\def\TU{H_{\gastep}}
\def\TUc{H_{\gastep,c}}

\def\rker{R_{\gastep}}
\def\sgP{P}

\def\lyapa{V_{\a}}
\def\lyape{V_{\ae}}

\def\rhoa{\rho_{a}}
\def\Ca{C_{a}}
\def\rhoae{\rho_{\ae/2}}
\def\Cae{C_{\ae/2}}
\def\kga{k_\gastep}
\def\rga{r_\gastep}
\def\qga{q_\gastep}

\def\XE{X}
\def\XEL{\Ybar}
\def\XL{Y}
\def\ZE{Z}

\def\Ybar{\overline{Y}}

\def\mubar{\overline{\mu}}

\def\be{b}
\def\a{a}
\def\ce{c}
\def\MM{M}

\def\diag{\operatorname{diag}}
\def\psig{x}
\def\nablat{\tilde{\nabla}}
\def\pginz{p}

\newcommandx{\bL}[1][1=]{\ifthenelse{\equal{#1}{}}{b}{b_{#1}} }
\def\sPoi{\phi}

\def\sPoihat{\hat{\phi}}
\def\sPoiN{\sPoi_N}
\def\inc{\delta}
\def\tf{f}
\def\tfb{\bar{f}}
\def\DD{\operatorname{D}}
\def\rr{r}
\def\nZ{Z}
\def\SS{S}
\def\MM{M}

\newcommandx{\bp}[1][1=]{\ifthenelse{\equal{#1}{}}{b_{p}}{b_{#1}}}

\def\Cpoly{\mathrm{C}_{\mathrm{poly}}}
\newcommandx{\polyga}[1][1=]{\ifthenelse{\equal{#1}{}}{\mathsf{P}_{\gastep}}{\mathsf{P}_{\gastep, #1}}}
\newcommandx{\polyp}[1][1=]{\ifthenelse{\equal{#1}{}}{\mathsf{P}}{\mathsf{P}_{#1}}}
\newcommandx{\polygaq}[1][1=]{\ifthenelse{\equal{#1}{}}{\mathsf{Q}_\gastep}{\mathsf{Q}_{\gastep, #1}}}

\def\gastep{\gamma}
\def\gaStep{\gamma}

\def\tildeVx{\tilde{V}_x}

\def\nablaS{\nabla^2}

\def\N{N}

\def\eventA{\mathsf{A}}
\def\varx{x}
\def\testfunh{h}

\def\xstar{x^\star}
\def\Tr{\operatorname{T}}

\newcommandx{\functionspace}[2][1=+]{\mathbb{F}_{#1}(#2)}
\newcommandx{\VarDeux}[3][3=]{\operatorname{Var}^{#3}_{#1}\left\{#2 \right\}}

\newcommand{\1}{\mathbbm{1}}

\newcommand{\Borel}{\mathcal{B}}

\newcommand{\LeftEqNo}{\let\veqno\@@leqno}

\newcommand{\floor}[1]{\left\lfloor #1 \right\rfloor}
\newcommand{\ceil}[1]{\left\lceil #1 \right\rceil}

\newcommand{\PE}{\mathbb{E}}
\newcommand{\PP}{\mathbb{P}}

\newcommand{\absolute}[1]{\left\vert #1 \right\vert}
\newcommand{\abs}[1]{\left\vert #1 \right\vert}
\newcommand{\tvnorm}[1]{\| #1 \|_{\mathrm{TV}}}

\newcommandx{\Vnorm}[2][1=V]{\| #2 \|_{#1}}
\newcommandx{\VnormEq}[2][1=V]{\left\| #2 \right\|_{#1}}
\newcommandx{\norm}[2][1=]{\ifthenelse{\equal{#1}{}}{\left\Vert #2 \right\Vert}{\left\Vert #2 \right\Vert^{#1}}}

\newcommandx{\normLigne}[2][1=]{\ifthenelse{\equal{#1}{}}{\Vert #2 \Vert}{\Vert #2\Vert^{#1}}}

\newcommand{\parenthese}[1]{\left(#1 \right)}

\newcommand{\defEns}[1]{\left\lbrace #1 \right\rbrace }
\newcommand{\defEnsLigne}[1]{\lbrace #1 \rbrace }

\newcommand{\ps}[2]{\left\langle#1,#2 \right\rangle}

\newcommandx\probaMarkovTilde[2][2=]
{\ifthenelse{\equal{#2}{}}{{\widetilde{\mathbb{P}}_{#1}}}{\widetilde{\mathbb{P}}_{#1}\left[ #2\right]}}

\newcommand{\expe}[1]{\PE \left[ #1 \right]}

\newcommand{\expeMarkov}[2]{\PE_{#1} \left[ #2 \right]}

\newcommand{\couplage}[2]{\Pi(#1,#2)}

\newcommand{\filtration}{\mathcal{F}}

\newcommand{\plusinfty}{+\infty}

\newcounter{hypoconbis}
\newcounter{saveconbis}
\newcommand\debutH{\begin{list}
{\textbf{H\arabic{hypoconbis}}}{\usecounter{hypoconbis}}\setcounter{hypoconbis}{\value{saveconbis}}}
\newcommand\finH{\end{list}\setcounter{saveconbis}{\value{hypoconbis}}}

\def\ie{\textit{i.e.}}
\def\eqsp{\;}
\newcommand{\coint}[1]{\left[#1\right)}
\newcommand{\ocint}[1]{\left(#1\right]}
\newcommand{\ooint}[1]{\left(#1\right)}
\newcommand{\ccint}[1]{\left[#1\right]}

\newcommandx{\weight}[2][2=n]{\omega_{#1,#2}^N}

\newcommand{\boule}[2]{\operatorname{B}(#1,#2)}
\newcommand{\boulefermee}[2]{\overline{\operatorname{B}}(#1,#2)}

\def\rmd{\mathrm{d}}
\newcommandx\sequence[3][2=,3=]
{\ifthenelse{\equal{#3}{}}{\ensuremath{\{ #1_{#2}\}}}{\ensuremath{\{ #1_{#2}, \eqsp #2 \in #3 \}}}}
\newcommandx{\sequencen}[2][2=n\in\N]{\ensuremath{\{ #1, \eqsp #2 \}}}
\newcommandx\sequenceDouble[4][3=,4=]
{\ifthenelse{\equal{#3}{}}{\ensuremath{\{ (#1_{#3},#2_{#3}) \}}}{\ensuremath{\{  (#1_{#3},#2_{#3}), \eqsp #3 \in #4 \}}}}
\newcommandx{\sequencenDouble}[3][3=n\in\N]{\ensuremath{\{ (#1_{n},#2_{n}), \eqsp #3 \}}}
\newcommand{\wrt}{w.r.t.}

\def\iid{i.i.d.}
\def\rme{\mathrm{e}}

\def\eg{e.g.}

\def\rset{\mathbb{R}}

\def\nset{\mathbb{N}}

\newcommandx{\CPE}[3][1=]{{\mathbb E}^{#3}_{#1}\left[#2 \right]} %%%% esperance conditionnelle
\newcommandx{\CPVar}[3][1=]{\mathrm{Var}^{#3}_{#1}\left\{ #2 \right\}}
\newcommand{\CPP}[3][]
{\ifthenelse{\equal{#1}{}}{{\mathbb P}\left(\left. #2 \, \right| #3 \right)}{{\mathbb P}_{#1}\left(\left. #2 \, \right | #3 \right)}}

\def\generator{\mathscr{A}}

\newcommandx{\osc}[2][1=]{\mathrm{osc}_{#1}(#2)}

\newcommand{\chunk}[4][]%
{\ifthenelse{\equal{#1}{}}{\ensuremath{{#2}_{#3:#4}}}{\ensuremath{#2^#1}_{#3:#4}}
}

\def\KL{\operatorname{KL}}
\def\Id{\operatorname{Id}}

\newcommand\density[2]{\frac{\rmd #2}{\rmd #1}}
\newcommand\densityLigne[2]{\rmd #2/\rmd #1}

\def\Csetfunction{\mathrm{C}}

\def\VStar{V}

\def\boreleanA{\mathrm{A}}

%% file: introduction.tex
\section{Introduction}
\label{sec:introduction}

The Unadjusted Langevin Algorithm (ULA) first introduced in the physics literature by \cite{parisi:1981} and popularized in the computational statistics community by \cite{grenander:1983} and \cite{grenander:miller:1994} is a technique to sample complex and high-dimensional probability distributions. This issue has far-reaching consequences in Bayesian statistics and machine learning \cite{andrieu:defreitas:doucet:jordan:2003}, \cite{cotter:roberts:stuart:white:2013}, aggregation of estimators \cite{dalalyan:tsybakov:2012} and molecular dynamics \cite{lelievre_stoltz_2016}.
More precisely, let $\invpi$ be a probability distribution on $\rset^d$
which has density (also denoted by $\invpi$) with respect to the
Lebesgue measure given for all $x\in\rset^d$ by,
\begin{equation*}
  \pi(x) = \rme^{-U(x)} \big/ \int_{\rset^d} \rme^{-U(y)} \rmd y \eqsp, \quad \text{with} \quad \int_{\rset^d} \rme^{-U(y)} \rmd y < \plusinfty \eqsp.
\end{equation*}
Assuming that $\pU:\rset^d\to\rset$ is continuously differentiable,
the overdamped Langevin stochastic differential equation (SDE)
associated with $\pi$ is given by
\begin{equation}
\label{eq:langevin}
  \rmd \XL_t = -\nabla U(\XL_t) \rmd t + \sqrt{2} \rmd B_t \eqsp,
\end{equation}
where $(B_t)_{t\geq 0}$ is a $d$-dimensional Brownian motion.
The discrete time Markov chain associated with the ULA algorithm is obtained by the Euler-Maruyama discretization scheme of the Langevin SDE
defined for $k\in\nset$ by,
\begin{equation}
\label{eq:def-euler}
  \XE_{k+1} = \XE_k - \gaStep \nablaU(\XE_k) + \sqrt{2\gaStep} \ZE_{k+1} \eqsp, \eqsp \XE_0 = x_0 \eqsp,
\end{equation}
where $x_0\in\rset^d$, $\gastep>0$ and $(\ZE_k)_{k\in\nset}$ are \iid~standard $d$-dimensional Gaussian variables.
Under adequate assumptions on a globally Lipschitz $\nablaU$, non-asymptotic bounds in total variation and Wasserstein distances between the distribution of $(\XE_k)_{k\in\nset}$ and $\invpi$ can be found in \cite{dalalyan:2017}, \cite{durmus2017}, \cite{2016arXiv160501559D}.
However, the ULA algorithm is unstable if $\nablaU$ is superlinear \ie~$\liminf_{\norm{x}\to\plusinfty} \norm{\nablaU(x)}/\norm{x} = \plusinfty$, see \cite[Theorem 3.2]{roberts:tweedie:1996}, \cite{mattingly:stuart:higham:2002} and \cite{Hutzenthaler1563}.
This is illustrated with a particular example in \cite[Lemma 6.3]{mattingly:stuart:higham:2002} where, the SDE \eqref{eq:langevin} is considered in one dimension with $\pU(x)=x^4/4$ along with the associated Euler discretization \eqref{eq:def-euler} and it is shown that for all $\gastep>0$, if $\expe{\XE_0^2} \geq 2 / \gastep$, one obtains $\lim_{n\to\plusinfty}\expe{X_n^2} = \plusinfty$. Moreover, the sample path $(X_n)_{n\in\nset}$ diverges to infinity with positive probability.

Until recently, either implicit numerical schemes, e.g. see \cite{mattingly:stuart:higham:2002} and \cite{HMS}, or adaptive stepsize schemes, e.g. see \cite{lamba:mattingly:stuart:2007}, were used to address this problem. However, in the last few years, a new generation of explicit numerical schemes, which are computationally efficient, has been introduced by ``taming'' appropriately the superlinearly growing drift, see \cite{hutzenthaler2012} and \cite{sabanis2013} for more details.
%This methodology has been extended to Milstein-type schemes, see \cite{KS17} and \cite{wanga2013}, which achieve optimal rate of convergence 1 and include the class of tamed Euler approximations for SDEs with constant diffusion coefficients. A more refined methodology in the latter direction appears in \cite{KS} which inspires some of the weak assumptions used in this article.

%To circumvent the problem, several strategies have been suggested and analyzed, see \cite{doi:10.1137/120902318} for an overview of different methods. %More precisely,
%The proposed solutions may be classified in three main categories:
%implicit schemes \cite{mattingly:stuart:higham:2002}, \cite{doi:10.1080/17442508.2011.651213}, \cite{FOROUSHBASTANI20121903},
%adaptive stepsize \cite{lamba:mattingly:stuart:2007}, \cite{MAO2015370}, \cite{2016arXiv161004003K}, \cite{ZHANG20171}
%and tamed schemes \cite{hutzenthaler2012}, \cite{sabanis2013}, \cite{hutzenthaler2015numerical}.

Nonetheless, with the exception of \cite{mattingly:stuart:higham:2002}, these works focus on the discretization of SDEs with superlinear coefficients in finite time.
%The same issue occurs when considering the ULA algorithm in the purpose of sampling from $\invpi$.
We aim at extending these techniques to sample from $\invpi$, the invariant measure of \eqref{eq:langevin}.
To deal with the superlinear nature of $\nablaU$,
we introduce a family of drift functions $(\dr)_{\gastep> 0}$ with $\dr:\rset^d\to\rset^d$ indexed by the step size $\gastep$ which are close approximations of $\nablaU$ in a sense made precise below.
Consider then the following Markov chain $(\XE_k)_{k\in\nset}$ defined for all $k \in\nset$ by
\begin{equation}
\label{eq:def_tamed_euler_1}
  \XE_{k+1} = \XE_k - \gaStep \Ugad(\XE_k) + \sqrt{2\gaStep} \ZE_{k+1} \eqsp, \eqsp \XE_0 = x_0 \eqsp.
\end{equation}
%Under appropriate assumptions on $\dr$ and for a Lyapunov function $V:\rset^d\to\coint{1,\plusinfty}$, we show that $(\XE_k)_{k \in\nset}$ is $V$-geometrically ergodic of invariant measure $\invpi_\gastep$ close to $\invpi$.
%In particular, the assumptions on $\dr$ are satisfied if $\dr$ is equal for all $x \in \rset^d$ to,
We suggest two different explicit choices for the family $(\dr)_{\gastep> 0}$ based on previous studies on the tamed Euler scheme \cite{hutzenthaler2012}, \cite{sabanis2013}, \cite{hutzenthaler2015numerical}. Define for all $\gastep>0$, $\TU, \TUc:\rset^d\to\rset^d$ for all $x\in\rset^d$ by
\begin{equation}
\label{eq:def-TU-TUc}
\TU (x) = \frac{\nabla U(x)}{1+ \gaStep \norm{\nabla U(x)}} \quad \text{and} \quad
\TUc (x) = \parenthese{\frac{\partial_i \pU(x)}{1+\gastep \absolute{\partial_i \pU(x)}}}_{i\in\defEns{1,\ldots,d}} \eqsp,
\end{equation}
where $\partial_i \pU$ is the $i^{\text{th}}$-coordinate of $\nablaU$.
The Euler scheme \eqref{eq:def_tamed_euler_1} with $\dr=\TU$, respectively $\dr=\TUc$, is referred to as the Tamed Unadjusted Langevin Algorithm (TULA), respectively the coordinate-wise Tamed Unadjusted Langevin Algorithm (TULAc).

Another line of work has focused on
the Metropolis Adjusted Langevin Algorithm (MALA) that consists in adding a Metropolis-Hastings step to the ULA algorithm. \cite{rabee:hairer:2013} provides a detailed analysis of MALA in the case where the drift coefficient is superlinear. Note also that a normalization of the gradient was suggested in \cite[Section 1.4.3]{roberts:tweedie:1996} calling it MALTA (Metropolis Adjusted Langevin Truncated Algorithm) and analyzed in \cite{Atchade2006} and \cite{CPA:CPA20306}.

The article is organized as follows. In \Cref{sec:conv-general},
the Markov chain $(X_k)_{k\in\nset}$ defined by \eqref{eq:def_tamed_euler_1} is shown to be $V$-geometrically ergodic \wrt~an invariant measure $\invpig$.
Non-asymptotic bounds between the distribution of $(X_k)_{k\in\nset}$ and $\invpi$ in total variation and Wasserstein distances are provided, as well as weak error bounds. In \Cref{sec:examples},
%explicit tamed drifts $\dr$ are given and
the methodology is illustrated through numerical examples. Finally, proofs of the main results appear in \Cref{sec:proofs}.

%\cite{LuPerHas2016a}
%\cite{livingstone:2017} 

%% file: notations.tex
\section*{Notations}

%Let $f : \rset^d \to \rset$ be a Lipschitz function, namely there exists $C \geq 0$ such that for all $x,y \in \rset^d$, $ \abs{f(x) - f(y)} \leq C \norm{x-y}$. Then we denote $\norm{f}_{\Lip} = \inf \{  \abs{f(x) - f(y)} \norm[-1]{x-y} \ | \ x,y \in \rset^d , x \not = y \}$. The Monge-Kantorovich theorem (see \cite[Theorem 5.9]{VillaniTransport}) implies that for all $\mu,\nu$ probabilities measure on $\rset^d$,
%\begin{equation}
%\label{eq:duality_Monge_kanto}
%W_1(\mu,\nu) = \sup \defEns{\int_{\rset^d} f(x) \mu (\rmd x) - \int_{\rset^d} f(x) \nu (\rmd x) \ | \ f: \rset^d \to \rset \ ; \   \norm{f}_{\Lip} \leq 1}\eqsp.
%\end{equation}

% Measures

Let $\mathcal{B}(\rset^d)$ denote the Borel $\sigma$-field of $\rset^d$. Moreover, let $\Lspace^1(\mu)$ be the set of $\mu$-integrable functions for $\mu$ a probability measure on $(\rset^d, \Borel(\rset^d))$. Further, $\mu(\tf)=\int_{\rset^d} \tf(x) \rmd \mu(x)$ for an $\tf\in\Lspace^1(\mu)$.
Given a Markov kernel $R$ on $\rset^d$, for all $x\in\rset^d$ and $\tf$ integrable under $R(x,\cdot)$, denote by $R \tf(x) = \int_{\rset^d} \tf(y) R(x, \rmd y)$.
Let $V: \rset^d \to \coint{1,\infty}$ be a measurable function.
The $V$-total variation distance between $\mu$ and $\nu$ is defined as
%\begin{equation*}
$\Vnorm[V]{\mu-\nu} = \sup_{\absolute{f} \leq V} \abs{\mu(f) - \nu(f)}$.
%\eqsp.
%\end{equation*}
If $V = 1$, then $\Vnorm[V]{\cdot}$ is the total variation  denoted by $\tvnorm{\cdot}$.
%For $p \geq 1$, denote by $\mathrm{L}^p(\pi)$ the set of  measurable functions such that $\pi(|f|^p) < \infty$.
Let $\mu$ and $\nu$ be two probability measures on a state space $\Omega$ with a given $\sigma$-algebra. If $\mu \ll \nu$, we denote by $\densityLigne{\nu}{\mu}$ the Radon-Nikodym derivative of $\mu$ \wrt~$\nu$. In that case, the Kullback-Leibler divergence of $\mu$ \wrt~to $\nu$ is defined as
\begin{equation*}
\KL(\mu | \nu) = \int_{\Omega} \density{\nu}{\mu} \log\parenthese{\density{\nu}{\mu}} \rmd \nu \eqsp.
\end{equation*}

We say that $\zeta$ is a transference plan of $\mu$ and $\nu$ if it is a probability measure on $(\rset^d \times \rset^d, \mathcal{B}(\rset^d \times \rset^d) )$  such that for any Borel set $\boreleanA$ of $\rset^d$, $\zeta(\boreleanA \times \rset^d) = \mu(\boreleanA)$ and $\zeta(\rset^d \times \boreleanA) = \nu(\boreleanA)$. We denote by $\couplage{\mu}{\nu}$ the set of transference plans of $\mu$ and $\nu$. Furthermore, we say that a couple of $\rset^d$-random variables $(X,Y)$ is a coupling of $\mu$ and $\nu$ if there exists $\zeta \in \couplage{\mu}{\nu}$ such that $(X,Y)$ are distributed according to $\zeta$.  For two probability measures $\mu$ and $\nu$, we define the Wasserstein distance of order $p \geq 1$ as
\begin{equation*}
%\label{eq:definition_wasserstein}
W_p(\mu,\nu) = \left( \inf_{\zeta \in \couplage{\mu}{\nu}} \int_{\rset^d \times \rset^d} \norm[p]{x-y}\rmd \zeta (x,y)\right)^{1/p} \eqsp.
\end{equation*}
By \cite[Theorem 4.1]{VillaniTransport}, for all $\mu,\nu$ probability measure on $\rset^d$, there exists a transference plan $\zeta^\star \in \couplage{\mu}{\nu}$ such that for any coupling $(X,Y)$ distributed according to $\zeta^\star$, $W_p(\mu,\nu) = \PE[\norm[p]{X-Y}]^{1/p}$.
%This kind of transference plan (respectively coupling) will be called an optimal transference plan (respectively optimal coupling) associated with $W_p$.  We denote by $\setProba_p(\rset^d)$ the set of probability measures with finite $p$-moment:  for all $\mu \in \setProba_p(\rset^d)$, $\int_{\rset^d} \norm[p]{x} \mu(\rmd x) < \plusinfty$. By \cite[Theorem 6.16]{VillaniTransport}, $\setProba_p(\rset^d)$ equipped with the Wasserstein distance $W_p$ of order $p$  is a complete separable metric space.

% Functions, vectors

For $u,v\in\rset^d$, define the scalar product $\ps{u}{v} = \sum_{i=1}^d u_i v_i$ and the Euclidian norm $\norm{u} = \ps{u}{u}^{1/2}$.
Denote by $\mathbb{S}^{d-1} = \defEns{u\in\rset^d : \norm{u} = 1}$.
For $k \in\nset$, $m,m' \in\nset^*$ and $\Omega,\Omega'$ two open sets of $\rset^m, \rset^{m'}$ respectively, denote by $\Csetfunction^k(\Omega, \Omega')$, the set of
$k$-times continuously differentiable functions. For $f \in
\Csetfunction^2(\rset^d, \rset)$, denote by $\nabla f$ the gradient of $f$, $\partial_i f$ the $i^{\text{th}}$-coordinate of $\nabla f$, $\Delta f$ the Laplacian of $f$ and $\nablaS f$ the Hessian of $\tf$. Define then for $x\in\rset^d$, $\norm{\nablaS \tf(x)} = \sup_{u\in \mathbb{S}^{d-1}} \norm{\nablaS \tf(x) u}$.
For $k\in\nset$ and $\tf \in\Csetfunction^k(\rset^d, \rset)$, denote by $\DD^i \tf$ the $i$-th derivative of $\tf$ for $i\in\defEns{0,\ldots,k}$, \ie~$\DD^i \tf$ is a symmetric $i$-linear map defined for all $x\in\rset^d$ and $j_1,\ldots,j_i\in\defEns{1,\ldots,d}$ by $\DD^i \tf(x) [e_{j_1}, \ldots, e_{j_i}] = \partial_{j_1\ldots j_i} \tf(x)$ where $e_1,\ldots,e_d$ is the canonical basis of $\rset^d$.
For $x\in\rset^d$ and $i\in\defEns{1,\ldots,k}$, define $\norm{\DD^0 \tf(x)} = \absolute{\tf(x)}$, $\norm{ \DD^i \tf (x)} = \sup_{u_1,\ldots,u_i \in\mathbb{S}^{d-1}} \DD^i \tf(x)[u_1, \ldots, u_i]$. Note that $\norm{\DD^1 \tf(x)} = \norm{\nabla \tf(x)}$ and $\norm{\DD^2 \tf(x)} = \norm{\nablaS \tf(x)}$. For $m,m'\in\nset^*$, define
\begin{multline*}
\Cpoly(\rset^m, \rset^{m'}) = \Big\{f \in \Csetfunction(\rset^m, \rset^{m'}) | \exists C_q, q \geq 0, \forall x \in \rset^m, \\ \norm{f(x)} \leq C_q (1+\norm[q]{x}) \Big\} \eqsp.
\end{multline*}

%For $a\in\rset_+$, $\floor{a}$ and $\ceil{a}$ denote respectively the floor and ceil functions evaluated in $a$.
%For $u\in\rset^d$, $\diag(u)$ is a diagonal matrix in $\rset^{d\times d}$ of diagonal $u$.
For all $x \in \rset^d$ and $M >0$, we denote by $\boule{x}{M}$ (respectively $\boulefermee{x}{M}$), the open (respectively close) ball centered at $x$ of radius $M$.
In the sequel, we take the convention that for $n,p \in \nset$, $n <p$ then $\sum_{p}^n =0$ and $\prod_p ^n = 1$. 

%% file: part1.tex
\section{Ergodicity and convergence analysis}
\label{sec:conv-general}
\begin{table}
  \centering
  \begin{tabular}{|c|c|c|}
    \hline
    distance & order of the upper bound & assumptions \\
    \hline
    $\VnormEq[V^{1/2}]{\delta_x \rker^n - \invpi}$ & $n\gastep \lambda^{n\gastep} V(x) + \sqrt{\gastep}$ & \Cref{assum:dr-close-nablaU}, \Cref{assum:ergodicity-drift}, \Cref{assum:grad-loc-lipschitz} and \Cref{assum:ergodicity-U} \\
    \hline
    $W_2^2 (\delta_x \rker^n, \invpi)$ & $n\gastep \lambda^{n\gastep} V(x) + \gastep$ & \Cref{assum:dr-close-nablaU}, \Cref{assum:ergodicity-drift},  \Cref{assum:grad-loc-lipschitz}, \Cref{assum:ergodicity-U} and \Cref{assum:convex} \\
    \hline
    $W_2^2 (\delta_x \rker^n, \invpi)$ & $n\gastep^{1+\betah} \lambda^{n\gastep} V(x) + \gastep^{1+\betah}$ & \Cref{assum:dr-close-nablaU}, \Cref{assum:ergodicity-drift},   \Cref{assum:ergodicity-U}, \Cref{assum:convex} and \Cref{assum:hessian} \\
    \hline
  \end{tabular}
  \caption{\label{table:rkergn-invpi-sum-up} Summary of the upper bounds on the distances between the distribution of the $n^{\text{th}}$ iteration of the Markov chain defined by \eqref{eq:def_tamed_euler_1} and $\invpi$.}
\end{table}
%\begin{table}
%  \centering
%  \begin{tabular}{|c|c|c|}
%    \hline
%    distance & upper bound & assumptions \\
%    \hline
%    $\VnormEq[\lyape^{1/2}]{\invpig - \invpi}$ & $C\gastep^{1/2}$ & \Cref{assum:grad-loc-lipschitz}, \Cref{assum:ergodicity-U}, \Cref{assum:dr-close-nablaU} and \Cref{assum:ergodicity-drift} \\
%    \hline
%    $W_2(\invpig, \invpi)$ & $C\gastep^{1/2}$ & \Cref{assum:dr-close-nablaU}, \Cref{assum:ergodicity-drift},  \Cref{assum:grad-loc-lipschitz}, \Cref{assum:ergodicity-U} and \Cref{assum:convex} \\
%    \hline
%    $W_2(\invpig, \invpi)$ & $C\gastep^{(1+\beta)/2}$ & \Cref{assum:dr-close-nablaU}, \Cref{assum:ergodicity-drift},   \Cref{assum:ergodicity-U}, \Cref{assum:convex} and \Cref{assum:hessian} \\
%    \hline
%  \end{tabular}
%  \caption{\label{table:invpig-invpi-sum-up} Summary of the upper bounds on the distances between $\invpig$ and $\invpi$.}
%\end{table}
In this Section, under appropriate assumptions on $\nablaU$ and $\dr$,
we show that the diffusion process $(\XL_t)_{t\geq 0}$ defined by \eqref{eq:langevin} and its discretization $(X_k)_{k\in\nset}$ defined by \eqref{eq:def_tamed_euler_1} satisfy a Foster-Lyapunov drift condition and are $V$-geometrically ergodic, see \Cref{prop:existence_ergodicity} and \Cref{propo:drift_tamed_euler}.
Second, for all $k \in \nset^*$, non-asymptotic bounds in $V$-norm  between the distribution of $X_k$
%defined by \eqref{eq:def_tamed_euler_1}
and $\invpi$ are established.
%where $V$ is an explicit Lyapunov function.
Our next results give non-asymptotic bounds in  Wasserstein
distance of order $2$, under the additional assumption that $\pU$ is
strongly convex. A summary
of our main contributions  is given in \Cref{table:rkergn-invpi-sum-up},
where $\lambda\in\coint{0,1}$.  We conclude this part by non-asymptotic
bounds on the bias and the variance of the ergodic average $n^{-1} \sum_{k=0}^{n-1} f(X_k)$, $n \in \nset^*$, used as an estimator of $\pi(f)$, for $f :\rset^d\to \rset$ sufficiently smooth.

Henceforth, it is assumed that $\pU$ is continuously differentiable.
Consider the following assumptions on $\pU$.

\begin{assumption}
\label{assum:grad-loc-lipschitz}
%$\pU$ is continuously differentiable and
There exist $\ell, \LG\in\rset_+$ such that for all $x,y \in \rset^d$,
\begin{equation*}
  \norm{  \nabla U(x) - \nabla U(y)}  \leq \LG \defEns{ 1 + \norm[\ell]{x} + \norm[\ell]{y}} \norm{x-y}\eqsp.
\end{equation*}
\end{assumption}

\begin{assumption}
\label{assum:ergodicity-U}
%$\pU$ is continuously differentiable and,
\begin{enumerate}[label=\roman*)]
\item
\label{item:ergodicity-U-1}
$\liminf_{\norm{x}\to \plusinfty} \norm{\nablaU(x)} = \plusinfty$.
\item
\label{item:ergodicity-U-2}
$\liminf_{\norm{x}\to \plusinfty} \ps{\frac{x}{\norm{x}}}{\frac{\nablaU(x)}{\norm{\nablaU(x)}}} > 0$.
\end{enumerate}
\end{assumption}
%
%\Cref{assum:ergodicity-U} is a standard assumption to show the ergodicity of the Markov chain targeting $\invpi$, see \eg~\cite[equation (34)]{JARNER2000341} in the case of the Random Walk Metropolis Hastings algorithm.
%
%
Note that under \Cref{assum:ergodicity-U}, $\liminf_{\norm{x}\to\plusinfty} \pU(x) = \plusinfty$, $\pU$ has a minimum $\xstar$ and $\nablaU (\xstar) = 0$. Without loss of generality, it is assumed that $\xstar = 0$. It implies under \Cref{assum:grad-loc-lipschitz} that for all $x\in\rset^d$,
\begin{equation}
\label{eq:estimate_nabla_u_x}
  \norm{\nabla U(x)} \leq 2 \LG \defEns{1+\norm[\ell+1]{x}}\eqsp.
\end{equation}
Besides, under \Cref{assum:ergodicity-U}-\ref{item:ergodicity-U-2}, there exists $C\in\rset$ such that for all $x\in\rset^d$, $  \ps{-\nabla U(x)}{x} \leq C$. By \cite[Theorem 2.1]{meyn:tweedie:1993:III}, \cite[Chapter IV, Theorems 2.3, 3.1]{ikeda:watanabe:1989} and \cite[Theorem 2.1]{roberts:tweedie:1996}, \eqref{eq:langevin} has a unique strong solution denoted $(\XL_t)_{t\geq 0}$. By \cite[Section 5.4.C, Theorem 4.20]{karatzas:shreve:1991}, one constructs the associated  strongly Markovian semigroup $(\sgP_t)_{t\geq 0}$  given for all $t\geq 0$, $x\in\rset^d$ and $\eventA \in \Borel(\rset^d)$ by $\sgP_t(x, \eventA) = \expe{\1_{\eventA}(\XL_t) | \XL_0=x}$.
%$(\sgP_t)_{t\geq 0}$ is reversible \wrt~$\invpi$ and for all $x\in\rset^d$, $\lim_{t\to\plusinfty}\tvnorm{\sgP_t(x,\cdot) - \invpi} = 0$.
Consider the infinitesimal generator $\generator$  associated with \eqref{eq:langevin} defined for all $h
\in \Csetfunction^2(\rset^d)$ and $x\in\rset^d$ by
\begin{equation}
\label{eq:generator-langevin}
  \generator h(x) = -\ps{\nabla U(\varx)}{\nabla \testfunh(\varx)} + \Delta \testfunh(\varx) \eqsp,
\end{equation}
and for any $\a \in \rset_+^*$, define the Lyapunov function $\lyapa : \rset^d \to \coint{1,\plusinfty}$ for all $x \in \rset^d$ by
\begin{equation}
  \label{eq:lyap_tamed_euler}
  \lyapa(x) = \exp\parenthese{\a (1+\norm[2]{x})^{1/2}} \eqsp.
\end{equation}
%The following result shows that for all $a>0$, $\lyapa$ is a Lyapunov function for the diffusion $(\XL_t)_{t\geq 0}$ that satisfies a drift condition, and $(\XL_t)_{t\geq 0}$ is $\lyapa$-geometrically ergodic \wrt~$\invpi$.
Foster-Lyapunov conditions enable to control the moments of the diffusion process $(\XL_t)_{t\geq 0}$, see \eg~\cite[Section 6]{meyn:tweedie:1993:III} or \cite[Theorem 2.2]{roberts:tweedie:1996}.
%This methodology was applied for example in \cite[Lemma 5.2]{mattingly:stuart:higham:2002}.

\begin{proposition}
\label{prop:existence_ergodicity}
Assume \Cref{assum:grad-loc-lipschitz}, \Cref{assum:ergodicity-U} and let $a\in\rset^*_+$. There exists $\bL[\a]\in\rset_+$ (given explicitly in the proof) such that for all $x \in \rset^d$
\begin{equation}
\label{eq:drift-langevin}
  \generator \lyapa(x) \leq - \a \lyapa(x) + \a\bL[\a] \eqsp
\end{equation}
and
\begin{equation*}
\sup_{t\geq 0} \sgP_t \lyapa(x) \leq \lyapa(x) + \bL[\a] \eqsp.
%\text{where} \quad \sgP_t \lyapa(x) = \int_{\rset^d} \lyapa(y) \sgP_t(x,\rmd y) \eqsp.
\end{equation*}
Moreover, there exist $\Ca\in\rset_+$ and $\rhoa\in\coint{0,1}$ such that for all $t\in\rset_+$ and probability measures $\mu_0, \nu_0$ on $(\rset^d, \Borel(\rset^d))$ satisfying $\mu_0(\lyapa) + \nu_0(\lyapa) < \plusinfty$,
\begin{equation}
\label{eq:V-geom-ergo}
\VnormEq[\lyapa]{\mu_0 \sgP_t - \nu_0 \sgP_t} \leq \Ca \rhoa^t \VnormEq[\lyapa]{\mu_0 - \nu_0} \eqsp, \eqsp
\VnormEq[\lyapa]{\mu_0 \sgP_t - \invpi} \leq \Ca \rhoa^t \mu_0(\lyapa) \eqsp.
\end{equation}
\end{proposition}

\begin{proof}
The proof is postponed to \Cref{sec:proof-Lyapunov-langevin}.
\end{proof}

The Markov chain $(X_k)_{k\in\nset}$ defined in \eqref{eq:def_tamed_euler_1} is a discrete-time approximation of the diffusion $(\XL_t)_{t\geq 0}$. To control the total variation and Wasserstein distances of the marginal distributions of $(X_k)_{k\in\nset}$ and $(\XL_t)_{t\geq 0}$, it is necessary to assume that for $\gastep>0$ small enough, $\dr$ and $\nablaU$ are close. This is formalized by \Cref{assum:dr-close-nablaU}.
%The counterpart of \Cref{assum:ergodicity-U} for the Markov chain $(X_k)_{k\in\nset}$ is \Cref{assum:ergodicity-drift};
Under the additional assumption \Cref{assum:ergodicity-drift}, we obtain the stability and ergodicity of $(\XE_k)_{k\in\nset}$.
%\Cref{lemma:checkA1A2} below shows that both assumptions are satisfied for $\TU$ and $\TUc$ defined in \eqref{eq:def-TU-TUc} assuming \Cref{assum:grad-loc-lipschitz} and \Cref{assum:ergodicity-U}.

\begin{assumptionA}\label{assum:dr-close-nablaU}
For all $\gastep> 0$, $\dr$ is continuous.
There exist $\alphag \geq 0$, $\Cal<\plusinfty$ such that for all $\gastep> 0$ and $x\in\rset^d$,
\begin{equation*}
\norm{\dr(x) - \nablaU(x)} \leq \gastep \Cal \parenthese{1+\norm[\alphag]{x}} \eqsp.
\end{equation*}
\end{assumptionA}

Note that under \Cref{assum:grad-loc-lipschitz}, \Cref{assum:dr-close-nablaU} and by \eqref{eq:estimate_nabla_u_x}, we have for all $x\in\rset^d$
\begin{equation}
\label{eq:estimate-drift-step}
\norm{\dr(x)} \leq 2 \LG \defEns{1+\norm[\ell+1]{x}} + \gastep \Cal \parenthese{1+\norm[\alphag]{x}} \eqsp.
\end{equation}

\begin{assumptionA}\label{assum:ergodicity-drift}
For all $\gastep>0$,
$\liminf_{\norm{x}\to\plusinfty} \ps{\frac{x}{\norm{x}}}{\dr(x)} - \frac{\gastep}{2\norm{x}}\norm[2]{\dr(x)} >0$.
\end{assumptionA}

\begin{lemma}
\label{lemma:checkA1A2}
Assume \Cref{assum:grad-loc-lipschitz} and \Cref{assum:ergodicity-U}. Let $\gastep>0$ and $\dr$ be equal to $\TU$ or $\TUc$ defined in \eqref{eq:def-TU-TUc}. Then \Cref{assum:dr-close-nablaU} and \Cref{assum:ergodicity-drift} are satisfied.
\end{lemma}

\begin{proof}
The proof is postponed to \Cref{sec:proof-lemma-checkA1A2}.
\end{proof}

%Under \Cref{assum:grad-loc-lipschitz}, \Cref{assum:dr-close-nablaU} and by \eqref{eq:estimate_nabla_u_x}, we have for all $x\in\rset^d$
%\begin{equation}
%\label{eq:estimate-drift-step}
%\norm{\dr(x)} \leq 2 \LG \defEns{1+\norm[\ell+1]{x}} + \gastep \Cal \parenthese{1+\norm[\alphag]{x}} \eqsp.
%\end{equation}
The Markov kernel $  \rker$ associated with \eqref{eq:def_tamed_euler_1} is given for all $\gamma >0$ , $x \in \rset^d$ and $\eventA \in \Borel(\rset^d)$ by
\begin{equation}
\label{eq:definition_rker}
  \rker(x,\eventA) = (2 \uppi)^{-d/2}\int_{\rset^d} \1_{\eventA}\parenthese{x-\gaStep \Ugad(x) + \sqrt{2\gaStep} z} \rme^{-\norm[2]{z}/2} \rmd z \eqsp.
\end{equation}
We then obtain the counterpart of \Cref{prop:existence_ergodicity} for the Markov chain $(X_k)_{k\in\nset}$.

%The following proposition shows that $(X_k)_{k\in\nset}$ defined in \eqref{eq:def_tamed_euler_1} satisfies a drift condition.

\begin{proposition}
  \label{propo:drift_tamed_euler}
Assume \Cref{assum:grad-loc-lipschitz}, \Cref{assum:dr-close-nablaU},  \Cref{assum:ergodicity-drift} and let $\gamma \in \rset_+^*$. There
exist $\MM, \ae, \be \in\rset^*_+$ (given explicitly in the proof) satisfying for all $x\in \rset^d$
\begin{equation}
\label{eq:drift-rker}
\rker \lyape(x) \leq \rme^{-\ae^2 \gastep} \lyape(x)
+ \gastep \be \1_{\boulefermee{0}{\MM}}(x) \eqsp .
\end{equation}
%and we have for all $n\in\nset$,
%\begin{equation}
%\label{eq:drift-rker-iter-n}
%\rker^n \lyape(x) \leq \rme^{-\ae^2 n \gastep} \lyape(x) + (\be/\ae^2)\rme^{\ae^2 \gastep} \eqsp.
%\end{equation}
In addition, $\rker$ has a unique invariant measure $\invpig$, $\rker$ is $\lyape$-geometrically ergodic \wrt~$\invpig$.
%\ie~there exist $\Ce{\varsigma\ae}\in\rset_+$ and $\rhoe{\varsigma\ae}\in\ooint{0,1}$ such that for all $n\in\nset$ and probability measure $\mu_0$ on $(\rset^d, \Borel(\rset^d))$ satisfying $\mu_0(\lyape^\varsigma) <\plusinfty$,
%\begin{equation*}
%\VnormEq[\lyape^\varsigma]{\mu_0 \rker^n - \invpig} \leq \Ce{\varsigma\ae} \rhoe{\varsigma\ae}^n \mu_0(\lyape^\varsigma) \eqsp.
%\end{equation*}
\end{proposition}

\begin{proof}
The proof is postponed to \Cref{sec:proof-Lyapunov-discrete-Vae}.
\end{proof}

Note that a straightforward induction of \eqref{eq:drift-rker} gives for all $n\in\nset$ and $x\in\rset^d$,
\begin{equation*}
\rker^n \lyape(x) \leq \rme^{-n\ae^2 \gastep} \lyape(x) + \{(\be \gastep)(1-\rme^{-n\ae^2 \gastep})\}/(1-\rme^{-\ae^2 \gastep}) \eqsp.
\end{equation*}
Using $1-\rme^{-\ae^2 \gastep} = \int_0^\gastep \ae^2 \rme^{-\ae^2 t} \rmd t \geq \gastep \ae^2 \rme^{-\ae^2 \gastep }$, we get for all $n\in\nset$
\begin{equation}
\label{eq:drift-rker-iter-n}
\rker^n \lyape(x) \leq \rme^{-\ae^2 n \gastep} \lyape(x) + (\be/\ae^2)\rme^{\ae^2 \gastep} \eqsp.
\end{equation}
In the following result, we compare the discrete and continuous time processes $(X_k)_{k\in\nset}$ and $(\XL_t)_{t\geq 0}$ using Girsanov's theorem and Pinsker's inequality, see \cite{dalalyan:2017} and \cite[Theorem 10]{durmus2017} for similar arguments.

\begin{theorem}
\label{prop:convergence-Vnorm}
Assume \Cref{assum:grad-loc-lipschitz}, \Cref{assum:ergodicity-U}, \Cref{assum:dr-close-nablaU} and \Cref{assum:ergodicity-drift}.
Let $\gastep_0 >0$. There exist $C>0$ and $\lambda\in\ooint{0,1}$ such that for all $\gastep\in\ocint{0,\gastep_0}$, $x\in\rset^d$ and $n\in\nset$,
\begin{equation}
\label{eq:thm-conv-Vnorm-1}
\VnormEq[\lyape^{1/2}]{\delta_x \rker^n - \invpi} \leq
C\parenthese{n\gastep \lambda^{n\gastep} \lyape(x) + \sqrt{\gastep}} \eqsp,
\end{equation}
where $\ae$ is defined in \Cref{propo:drift_tamed_euler} and for all $\gastep\in\ocint{0,\gastep_0}$,
\begin{equation}
\label{eq:thm-conv-Vnorm-2}
\VnormEq[\lyape^{1/2}]{\invpig - \invpi} \leq C\sqrt{\gastep} \eqsp.
\end{equation}
\end{theorem}

\begin{proof}
The proof is postponed to \Cref{sec:proof-convergence-Vnorm}.
\end{proof}

%Note that the true dependence in the initial condition in \eqref{eq:thm-conv-Vnorm-1} is $\lyape^{1/2}(x)$ multiplied by some polynomial in $\norm{x}$. However, its expression is very involved and for sake of simplicity, we upper bound it by $\lyape(x)$.
By adding strong convexity for the potential, one obtains the corresponding bounds for  the Wasserstein distance of order $2$.
\begin{assumption}
\label{assum:convex}
$U$ is strongly convex, \ie~there exists $\mU >0$ such that for all $x,y \in \rset^d$,
\begin{equation*}
\ps{ \nabla U(x) - \nabla U(y)}{x-y}  \geq \mU \norm[2]{x-y} \eqsp.
\end{equation*}
\end{assumption}

%In this context, the assumption is classical, see \eg~\cite{eberle:2015},  \cite{2016arXiv160501559D}, \cite{dalalyan:2017}.
By coupling $(\XL_t)_{t\geq 0}$ and the linear interpolation of $(X_k)_{k\in\nset}$ with the same Brownian motion, the following result is obtained.
% and related to the curvature-dimension condition see \eg~\cite{VillaniTransport} for an introduction to this subject.

\begin{theorem}
\label{thm:wasserstein-gamma2}
Assume \Cref{assum:dr-close-nablaU}, \Cref{assum:ergodicity-drift},  \Cref{assum:grad-loc-lipschitz}, \Cref{assum:ergodicity-U} and \Cref{assum:convex}.
Let $\gastep_0>0$. There exist $C>0$ and $\lambda\in\ooint{0,1}$ such that for all $x\in\rset^d$, $\gastep\in\ocint{0,\gastep_0}$ and $n\in\nset$,
\begin{equation}
\label{eq:thm-wasserstein2-1}
W_2^2 (\delta_x \rker^n, \invpi) \leq
C \parenthese{n\gastep \lambda^{n\gastep} \lyape(x) + \gastep} \eqsp,
\end{equation}
where $\ae$ is defined in \Cref{propo:drift_tamed_euler} and for all $\gastep\in\ocint{0,\gastep_0}$,
\begin{equation}
\label{eq:thm-wasserstein2-2}
W_2^2(\invpig, \invpi) \leq C\gastep \eqsp.
\end{equation}
\end{theorem}

\begin{proof}
The proof is postponed to \Cref{sec:proof-strongly-convex}.
\end{proof}

If $\pU\in\Csetfunction^2(\rset^d,\rset)$ and under the following assumption on $\nablaS \pU$, the bound can be improved.

\begin{assumption}
\label{assum:hessian}
$\pU$ is twice continuously differentiable and there exist $\nue, \LH \in \rset_+$ and $\betah\in\ccint{0,1}$ such that for all $x,y\in\rset^d$,
\begin{equation*}
  \norm{\nabla^2 U(x) - \nabla^2 U(y)} \leq \LH \defEns{1+\norm[\nue]{x} +\norm[\nue]{y}} \norm[\betah]{x-y}\eqsp.
\end{equation*}
\end{assumption}

%If $\pU$ is twice continuously differentiable and there exist $\ell,\LG \geq 0 $ such that for all $x \in \rset^d$,
%\begin{equation}
%\label{lemma:prac-grad-loc-lipschitz}
%\norm{\nabla^2 U(x)} \leq \LG \defEns{1+ \norm[\ell]{x}} \eqsp,
%\end{equation}
%then \Cref{assum:grad-loc-lipschitz} is satisfied.
%Note that \Cref{assum:hessian} implies \Cref{assum:grad-loc-lipschitz} by \Cref{lem:reg-hessian-U-growth}-\ref{item:reg-hessian-U-growth-1} and \eqref{lemma:prac-grad-loc-lipschitz} with $\LG=\CH$ and $\ell = \nu+\betah$.
It is shown in \Cref{sec:proof-strongly-convex} that \Cref{assum:hessian} implies \Cref{assum:grad-loc-lipschitz}.
% with $\LG=\max(2\LH,\norm{\nablaS U(0)})$ and $\ell=\nu+\betah$.

\begin{theorem}
\label{thm:wasserstein-gamma3}
Assume \Cref{assum:dr-close-nablaU}, \Cref{assum:ergodicity-drift},   \Cref{assum:ergodicity-U}, \Cref{assum:convex} and \Cref{assum:hessian}.
Let $\gastep_0>0$. There exist $C>0$ and $\lambda\in\ooint{0,1}$ such that for all $x\in\rset^d$, $\gastep\in\ocint{0,\gastep_0}$ and $n\in\nset$,
\begin{equation}
\label{eq:thm-wasserstein3-1}
W_2^2 (\delta_x \rker^n, \invpi) \leq
C \parenthese{n\gastep^{1+\betah} \lambda^{n\gastep} \lyape(x) + \gastep^{1+\betah}} \eqsp,
\end{equation}
where $\ae$ is defined in \Cref{propo:drift_tamed_euler} and for all $\gastep\in\ocint{0,\gastep_0}$,
\begin{equation}
\label{eq:thm-wasserstein3-2}
W_2^2(\invpig, \invpi) \leq C\gastep^{1+\betah} \eqsp.
\end{equation}
\end{theorem}

\begin{proof}
The proof is postponed to \Cref{sec:proof-strongly-convex}.
\end{proof}

%\begin{assumption}
%\label{assumption:sol-poisson}
%There exist $\sPoi\in\Csetfunction^4(\rset^d,\rset)$, $C>0$ and $p\in\nset$ such that for all $x\in\rset^d$, $\generator \sPoi(x) = -\parenthese{\tf(x) -  \invpi(\tf)}$ and $\max_{i\in\{0,\ldots,4\}} \norm{\DD^i \sPoi(x)} \leq C(1+\norm[p]{x})$.
%\end{assumption}

The exponent of $\gastep$ in \eqref{eq:thm-wasserstein2-1} is improved from $1$ to $1+\betah$. In particular, if $\nablaS \pU$ is Lipschitz, $\nue=0$, $\betah=1$, and \cite[Theorem 8]{2016arXiv160501559D} is recovered.

Let $(X_k)_{k\in\nset}$ be the Markov chain defined in \eqref{eq:def_tamed_euler_1}. To study the empirical average $(1/n)\sum_{k=0}^{n-1} \{ \tf(X_k) - \invpi(\tf) \}$ for $n\in\nset^*$, we follow a method introduced in \cite{doi:10.1137/090770527} and based on the Poisson equation.
For $\tf$ a $\invpi$-integrable function, the Poisson equation associated with the generator $\generator$ defined in \eqref{eq:generator-langevin} is given for all $x\in\rset^d$ by
\begin{equation}
\label{eq:eq-Poisson}
\generator \sPoi(x) = -\parenthese{\tf(x) -  \invpi(\tf)} \eqsp,
\end{equation}
where $\sPoi$, if it exists, is a solution of the Poisson equation.
This equation has proved to be a useful tool to analyze additive functionals of diffusion processes, see \eg~\cite{cattiaux2012central} and references therein.
The existence and regularity of a solution of the Poisson equation has been investigated in \cite{glynn1996}, \cite{pardoux2001}, \cite{kopec:overdamped}, \cite{2016arXiv161106972G}.
In that purpose, the following additional assumption on $\pU$ is necessary.
%We first state a result on the solutions of the Poisson equation under appropriate assumptions on $\pU$ and $\tf$.

\begin{assumption}
\label{assumption:U-C4}
$\pU\in\Csetfunction^4(\rset^d,\rset)$ and $\norm{\DD^i \pU} \in \Cpoly(\rset^d, \rset_+)$ for $i\in\defEns{1,\ldots,4}$.
\end{assumption}

%\begin{proposition}
%\label{prop:existence-sol-Poisson}
%Assume \Cref{assum:ergodicity-U} and \Cref{assumption:U-C4}. Let $\tf\in\Csetfunction^3(\rset^d,\rset)$ be such that $\norm{\DD^i \tf} \in \Cpoly(\rset^d, \rset_+)$ for $i\in\defEns{0,\ldots,3}$.
%Then, there exists a solution of the Poisson equation \eqref{eq:eq-Poisson} $\sPoi\in\Csetfunction^4(\rset^d,\rset)$, such that $\norm{\DD^i \sPoi} \in\Cpoly(\rset^d,\rset_+)$  for $i\in\defEns{0,\ldots,4}$.
%\end{proposition}
%
%\begin{proof}
%The proof is postponed to \Cref{sec:proof-prop-sol-Poisson}.
%\end{proof}

\begin{theorem}
\label{prop:bias-MSE-poisson-equation}
Assume \Cref{assum:ergodicity-U}, \Cref{assumption:U-C4}, \Cref{assum:dr-close-nablaU} and \Cref{assum:ergodicity-drift}. Let $\tf\in\Csetfunction^3(\rset^d,\rset)$ be such that $\norm{\DD^i \tf} \in \Cpoly(\rset^d, \rset_+)$ for $i\in\defEns{0,\ldots,3}$.
Let $\gastep_0 > 0$ and $(\XE_k)_{k\in\nset}$ be the Markov chain defined by \eqref{eq:def_tamed_euler_1} and starting at $X_0=0$. There exists $C>0$ such that for all $\gastep\in\ocint{0,\gastep_0}$ and $n\in\nset^*$,
\begin{equation}\label{eq:bias-poisson-eq}
\absolute{\expe{\frac{1}{n}\sum_{k=0}^{n-1} \tf(\XE_k) - \invpi(\tf)}} \leq C \parenthese{\gastep + \frac{1}{n\gastep}}
\end{equation}
and
\begin{equation}\label{eq:MSE-poisson-eq}
\expe{\parenthese{\frac{1}{n}\sum_{k=0}^{n-1} \tf(\XE_k) - \invpi(\tf)}^2} \leq C \parenthese{\gastep^2 + \frac{1}{n\gastep}} \eqsp.
\end{equation}
\end{theorem}

\begin{proof}
The proof is postponed to \Cref{sec:proof-prop:bias-MSE-poisson-equation}.
\end{proof}

Note that the standard rates of convergence are recovered, see \cite[Theorems 5.1, 5.2]{doi:10.1137/090770527}.

%% file: examples.tex
\section{Numerical examples}
\label{sec:examples}

%In this Section, we suggest two different explicit formulations for $\dr$ based on previous studies on the tamed Euler scheme \cite{hutzenthaler2012}, \cite{sabanis2013}, \cite{hutzenthaler2015numerical}. Define for all $\gastep>0$, $\TU, \TUc:\rset^d\to\rset^d$ for all $x\in\rset^d$ by,
%\begin{equation}
%\label{eq:def-TU-TUc}
%\TU (x) = \frac{\nabla U(x)}{1+ \gaStep \norm{\nabla U(x)}} \eqsp, \quad
%\TUc (x) = \parenthese{\frac{\partial_i \pU(x)}{1+\gastep \absolute{\partial_i \pU(x)}}}_{i\in\defEns{1,\ldots,d}} \eqsp.
%\end{equation}
%The Euler scheme \eqref{eq:def_tamed_euler_1} with $\dr=\TU$, respectively $\dr=\TUc$, is referred to as the Tamed Unadjusted Langevin Algorithm (TULA), respectively the Tamed Unadjusted Langevin Algorithm coordinate-wise (TULAc). We first check that for both cases \Cref{assum:dr-close-nablaU} and \Cref{assum:ergodicity-drift} are satisfied.

%\begin{lemma}
%\label{lemma:checkA1A2}
%Assume \Cref{assum:grad-loc-lipschitz} and \Cref{assum:ergodicity-U}. Let $\gastep>0$ and $\dr$ be defined by $\TU$ or $\TUc$. Then \Cref{assum:dr-close-nablaU} and \Cref{assum:ergodicity-drift} are satisfied.
%\end{lemma}
%
%\begin{proof}
%The proof is postponed to \Cref{sec:proof-lemma-checkA1A2}.
%\end{proof}

We illustrate our theoretical results using three numerical examples.
%The methodology is tested against three numerical examples.

\paragraph*{Multivariate Gaussian variable in high dimension}
We first consider a multivariate Gaussian variable in dimension $d\in\{100,1000\}$ of mean $0$ and covariance matrix $\Sigma = \diag(1,\ldots,d)$. The potential $\pU:\rset^d \to \rset$ defined for all $x\in\rset^d$ by $\pU(x) = (1/2) x^{\Tr} \Sigma^{-1} x$ is $d^{-1}$-strongly convex and $1$-gradient Lipschitz. The assumptions \Cref{assum:grad-loc-lipschitz}, \Cref{assum:ergodicity-U}, \Cref{assum:convex}, \Cref{assum:hessian} with $\betah=1$ and \Cref{assumption:U-C4} are thus satisfied. Note that in this case, ULA is stable and the analysis of \cite{dalalyan:2017}, \cite{durmus2017}, \cite{2016arXiv160501559D} valid. Nevertheless, implementing TULA and TULAc on this example is still of interest. Indeed, some Bayesian posterior distributions have intricate expressions and identifying the superlinear part in the gradient $\nablaU$ may be a difficult task. Within this context, we check the robustness of TULA and TULAc with respect to (globally) Lipschitz $\nablaU$.

We also consider in \Cref{sec:suppl-badly-conditioned-gaussian} a badly conditioned multivariate Gaussian variable in dimension $d=100$ of mean $0$ and covariance matrix $\Sigma = \diag(10^{-5}, 1, \ldots, 1)$. In this example, ULA requires a step size of order $10^{-5}$ to be stable which implies a large number of iterations to obtain relevant results. On the other side, TULA and TULAc are applicable with a step size of order $10^{-2}$ and within a relatively small number of iterations, valid results for the axes $2$ to $100$ are obtained.

\paragraph*{Double well}
The potential is defined for all $x\in\rset^d$ by $\pU(x) = (1/4)\norm[4]{x} - (1/2)\norm[2]{x}$. We have $\nablaU(x) = (\norm[2]{x} -1) x$ and $\nablaS \pU(x) = (\norm[2]{x} -1) \Id + 2x x^{\Tr}$. We get $\norm{\nablaS \pU(x)} = 3\norm[2]{x} -1$, $\ps{x}{\nablaU(x)} = \norm{x}\norm{\nablaU(x)}$ for $\norm{x} \geq 1$ and
\[ \norm{\nablaS \pU(x) - \nablaS \pU(y)} \leq 3\parenthese{\norm{x}+\norm{y}}\norm{x-y} \eqsp, \]
so that \Cref{assum:grad-loc-lipschitz}, \Cref{assum:ergodicity-U},  \Cref{assum:hessian} with $\beta=1$ and \Cref{assumption:U-C4} are satisfied.

%We have for all $x,y\in\rset^d$,
%\begin{equation*}
%\nablaS \pU(x) - \nablaS \pU(y) = \parenthese{\norm[2]{x} - \norm[2]{y}} \Id + 2 \defEns{\parenthese{x-y} x^{\Tr} + y \parenthese{x-y}^{\Tr}} \eqsp,
%\end{equation*}
%and $\norm{\nablaS \pU(x) - \nablaS \pU(y)} \leq 3\parenthese{\norm{x}+\norm{y}}\norm{x-y}$ so that  is satisfied with $\beta=1$.

\paragraph*{Ginzburg-Landau model}
%The third example comes from the Ginzburg-Landau model \cite[Section 6.2]{livingstone:2017},
This model of phase transitions in physics \cite[Section 6.2]{livingstone:2017} is defined on a three-dimensional $d=\pginz^3$ lattice for $\pginz\in\nset^*$ and the potential is given for $\psig = \parenthese{\psig_{ijk}}_{i,j,k\in\defEns{1,\ldots,\pginz}} \in \rset^{d}$ by
\begin{equation*}
\pU(\psig) = \sum_{i,j,k=1}^{\pginz} \defEns{\frac{1-\tau}{2}\psig_{ijk}^2 + \frac{\tau \alpha}{2}\norm[2]{\nablat \psig_{ijk}} + \frac{\tau\lambda}{4}\psig_{ijk}^4} \eqsp,
\end{equation*}
where $\alpha,\lambda,\tau>0$ and $\nablat \psig_{ijk} = (\psig_{i_{+}jk}-\psig_{ijk},\psig_{ij_{+}k}-\psig_{ijk},\psig_{ijk_{+}}-\psig_{ijk})$ with $i_{\pm} = i\pm 1 \eqsp \text{mod} \eqsp \pginz$ and similarly for $j_{\pm}, k_{\pm}$.
In the simulations, $\pginz$ is equal to $10$.
We have
\begin{multline*}
\nablaU (\psig) = \Big\{
\tau\alpha\parenthese{6\psig_{ijk}-\psig_{i_{+}jk}-\psig_{ij_{+}k}-\psig_{ijk_{+}}-\psig_{i_{-}jk}-\psig_{ij_{-}k}-\psig_{ijk_{-}}} \\
+ (1-\tau) \psig_{ijk} + \tau\lambda \psig_{ijk}^3
\Big\}_{i,j,k\in\defEns{1,\ldots,\pginz}} \eqsp,
\end{multline*}
and
\begin{equation*}
\nablaS \pU(\psig) = \diag\parenthese{\parenthese{1-\tau + 6\tau\alpha + 3\tau\lambda \psig_{ijk}^2}_{i,j,k\in\defEns{1,\ldots,\pginz}}} + M \eqsp,
\end{equation*}
where $M\in\rset^{d \times d}$ is a constant matrix.
\Cref{assum:grad-loc-lipschitz}, \Cref{assum:hessian} with $\betah=1$ and \Cref{assumption:U-C4} are thus satisfied.
Using that $\psig \mapsto \sum_{i,j,k=1}^{\pginz} \norm[2]{\nablat \psig_{ijk}}$ is convex by composition of convex functions and its gradient evaluated in $0$ is $0$, we have for all $\psig\in\rset^d$,
\[ \ps{\psig}{\nablaU(\psig)} \geq \sum_{i,j,k=1}^{\pginz} \{(1-\tau) \psig_{ijk}^2 + \tau\lambda \psig_{ijk}^4 \} \eqsp. \]
By Cauchy-Schwarz inequality, $\defEns{\sum_{i,j,k=1}^{\pginz} \psig_{ijk}^2}^2 \leq d \sum_{i,j,k=1}^{\pginz} \psig_{ijk}^4$, and for all $x\in\rset^d$, $\norm[2]{\psig} \geq (2\absolute{1-\tau}d) / (\tau\lambda)$, we get $\ps{\psig}{\nablaU(\psig)} \geq \{(\tau\lambda)/2\} \sum_{i,j,k=1}^{\pginz} \psig_{ijk}^4$. Besides, we have
\[ \norm{\nablaU(\psig)} \leq (\absolute{1-\tau} + 12\tau\alpha) \norm{\psig} + \tau\lambda\norm{(\psig_{ijk}^3)_{i,j,k\in\defEns{1,\ldots,\pginz}}} \eqsp. \]
Let $a,b,c\in\defEns{1,\ldots,\pginz}$ be such that $\absolute{\psig_{abc}} = \max \absolute{\psig_{ijk}}$.
We get
\[ \norm{\psig} \norm{(\psig_{ijk}^3)_{i,j,k\in\defEns{1,\ldots,\pginz}}} \leq d \psig_{abc}^4 \leq d \sum_{i,j,k=1}^{\pginz} \psig_{ijk}^4 \eqsp. \]
Finally, for $\norm[2]{\psig} \geq \max\{1, (2\absolute{1-\tau}d) / (\tau\lambda)\}$, we obtain
\[ \norm{\psig}\norm{\nablaU(\psig)} \leq  \defEns{\frac{2d\absolute{1-\tau}}{\tau\lambda}+\frac{24\alpha d}{\lambda} + 2d} \ps{\psig}{\nablaU(\psig)} \eqsp, \]
and \Cref{assum:ergodicity-U} is satisfied.

We benchmark TULA and TULAc against ULA given by \eqref{eq:def-euler}, MALA and a Random Walk Metropolis-Hastings with a Gaussian proposal (RWM). TMALA (Tamed Metropolis Adjusted Langevin Algorithm) and TMALAc (coordinate-wise Tamed Metropolis Adjusted Langevin Algorithm), the Metropolized versions of TULA and TULAc, are also included in the numerical tests. Their theoretical analysis is similar to the one of MALTA \cite[Proposition 2.1]{Atchade2006}.

Since double well and Ginzburg-Landau models are coordinate-wise exchangeable, the results are provided only for their first coordinate.
The Markov chains associated with these models are started at $X_0=0, (10,0^{\otimes(d-1)}), (100,0^{\otimes(d-1)})$, $(1000,0^{\otimes(d-1)})$ and for the multivariate Gaussian at a random vector of norm $0,10,100,1000$.
For the Gaussian and double well examples, for each initial condition, algorithm, step size $\gastep\in\defEns{10^{-3},10^{-2},10^{-1}}$, we run $100$ independent Markov chains started at $X_0$ of $10^6$ samples (respectively $10^5$) in dimension $d=100$ (respectively $d=1000$).
For the Ginzburg-Landau model, we run $100$ independent Markov chains started at $X_0$ of $10^5$ samples.
For each run, we estimate the 1st and 2nd moment for the first and last coordinate, \ie~$\int_{\rset^d} x_i \invpi(x) \rmd x$ for $i\in\defEns{1,d}$, by the empirical average
and we compute the boxplots of the errors.
%on the 1st and 2nd moment for the first and the last coordinate, \ie~$\expe{X_i} - \invpi(X_i)$ and $\expe{X^2_i} - \invpi(X_i^2)$ for $i\in\defEns{1,d}$.
For ULA, if the norm of $X_k$ for $k\in\nset$ exceeds $10^5$, the chain is stopped and for this step size $\gastep$ the trajectory of ULA is not taken into account. For MALA, RWM, TMALA and TMALAc, if the acceptance ratio is below $0.05$, we similarly do not take into account the corresponding trajectories.

For the three examples and for $i\in\defEns{1,\ldots,d}$, $\int_{\rset^d} x_i \invpi(x) \rmd x = 0$. By symmetry, for the double well, we have for $i\in\defEns{1,\ldots,d}$ and $r\in\rset_+$,
\begin{equation*}
\expe{X_i^2} = d^{-1} \int_{\rset_+} r^2 \nu(r) \rmd r \Big/ \int_{\rset_+} \nu(r) \rmd r \eqsp,\quad
\nu(r) = r^{d-1} \exp\defEns{(r^2/2)-(r^4/4)} \eqsp.
\end{equation*}
A Random Walk Metropolis run of $10^7$ samples gives $\int_{\rset^d} x_i^2 \invpi(x) \rmd x \approx 0.104 \pm 0.001$ for $d=100$ and $\int_{\rset^d} x_i^2 \invpi(x) \rmd x \approx 0.032 \pm 0.001$ for $d=1000$.

Because of lack of space, we only display some boxplots in \Cref{figure:icg-1c-1st-d1000-x0,figure:dw-1st-d100-x100,figure:dw-2nd-d100-x0,figure:g-1st-d1000-x100}.
The Python code and all the figures are available at \url{https://github.com/nbrosse/TULA}.
We remark that TULA, TULAc and to a lesser extent, TMALA and TMALAc, have a stable behavior even with large step sizes and starting far from the origin.
This is particularly visible in \Cref{figure:dw-1st-d100-x100,figure:g-1st-d1000-x100} where ULA diverges (\ie~$\liminf_{k\to\plusinfty} \expe{\norm{X_k}} = \plusinfty$) and MALA does not move even for small step sizes $\gastep=10^{-3}$.
Note however the existence of a bias for ULA, TULA and TULAc in \Cref{figure:dw-2nd-d100-x0}.
Finally, comparison of the results shows that TULAc is preferable to TULA.

Note that other choices are possible for $\dr$, depending on the model under study. For example, in the case of the double well, we could "tame" only the superlinear part of $\nablaU$, \ie~consider for all $\gastep>0$ and $x\in\rset^d$,
\begin{equation}
\label{eq:smart-TULA}
\dr(x) = \frac{\norm[2]{x} x}{1+\gastep\norm[2]{x}} - x \eqsp.
\end{equation}
\Cref{assum:dr-close-nablaU} is satisfied and we have
\begin{align*}
&\ps{\frac{x}{\norm{x}}}{\dr(x)} - \frac{\gastep}{2\norm{x}}\norm[2]{\dr(x)} = \frac{\norm[3]{x}}{1+\gastep\norm[2]{x}} \Big\{ 1+\gastep - \frac{\gastep}{2}\frac{\norm[2]{x}}{1+\gastep\norm[2]{x}}\Big\} \\
&\phantom{------------------------}
- \norm{x}\defEns{1+(\gastep/2)} \eqsp, \\
&\liminf_{\norm{x}\to\plusinfty} \ps{\frac{x}{\norm[2]{x}}}{\dr(x)} - \frac{\gastep}{2\norm[2]{x}}\norm[2]{\dr(x)} = \frac{\gastep^{-1} - \gastep}{2} \eqsp.
\end{align*}
\Cref{assum:ergodicity-drift} is satisfied if and only if $\gastep\in\ooint{0,1}$.
It is striking to see that this theoretical threshold is clearly visible on the simulations.
The algorithm \eqref{eq:def_tamed_euler_1} with $\dr$ defined by \eqref{eq:smart-TULA} obtains similar results as TULAc for $\gastep<1$ but for $\gastep=1$, the algorithm diverges.
%see \Cref{figure:dw-1st-d100-x10-smart} where the algorithm \eqref{eq:def_tamed_euler_1} with $\dr$ defined by \eqref{eq:smart-TULA} is denoted by sTULA (for "smart" TULA).

Given the results of the numerical experiments, TULAc should be chosen over ULA to sample from general probability distributions. Indeed, TULAc has similar results as ULA when the step size is small and is more stable when using larger step sizes.

\begin{figure}
%\begin{center}
\hspace*{-2cm}\includegraphics[scale=0.55]{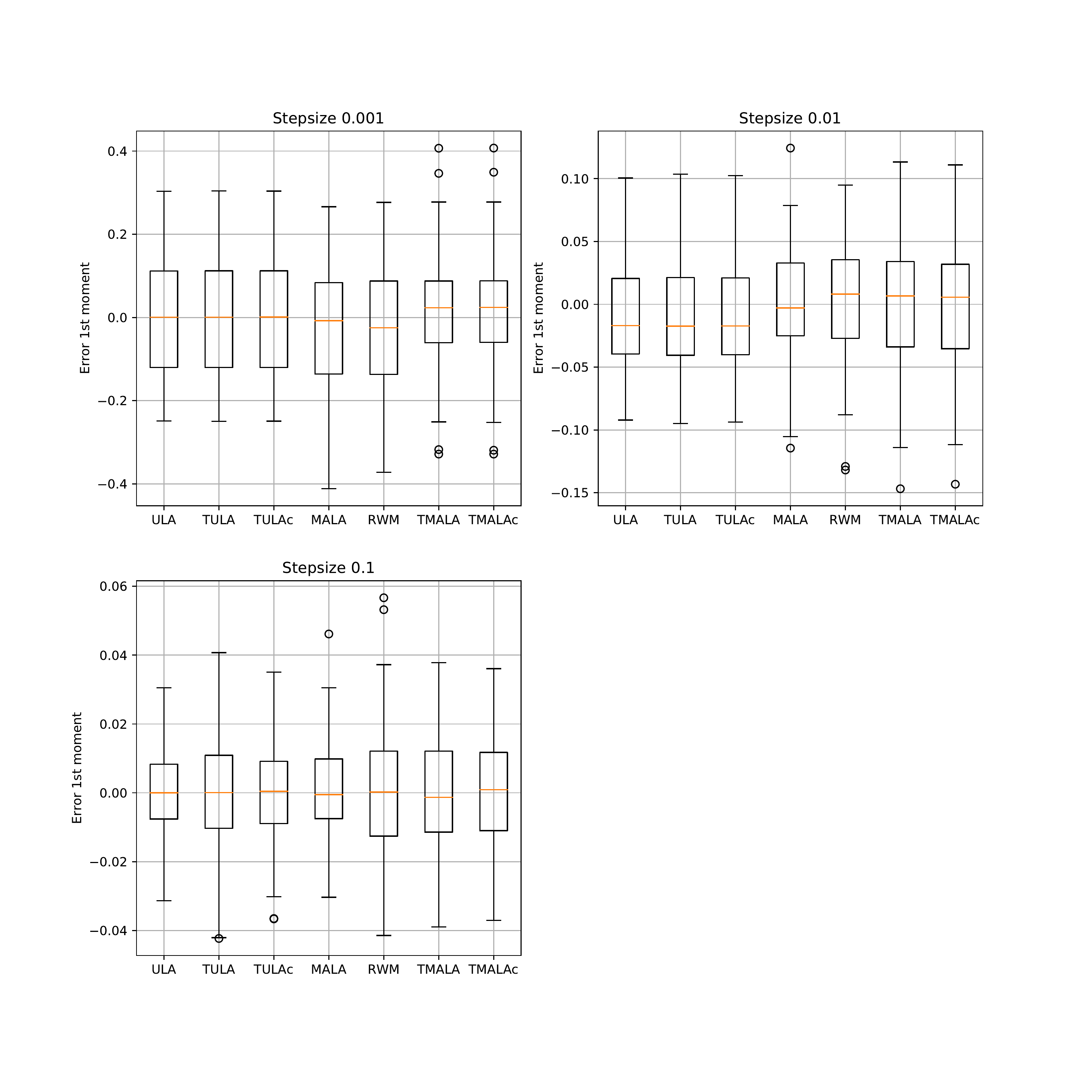}
%\end{center}
\caption{\label{figure:icg-1c-1st-d1000-x0} Boxplots of the error on the first moment for the multivariate Gaussian (first coordinate) in dimension $1000$ starting at $0$ for different step sizes.}
\end{figure}

\begin{figure}
%\begin{center}
\hspace*{-2cm}\includegraphics[scale=0.55]{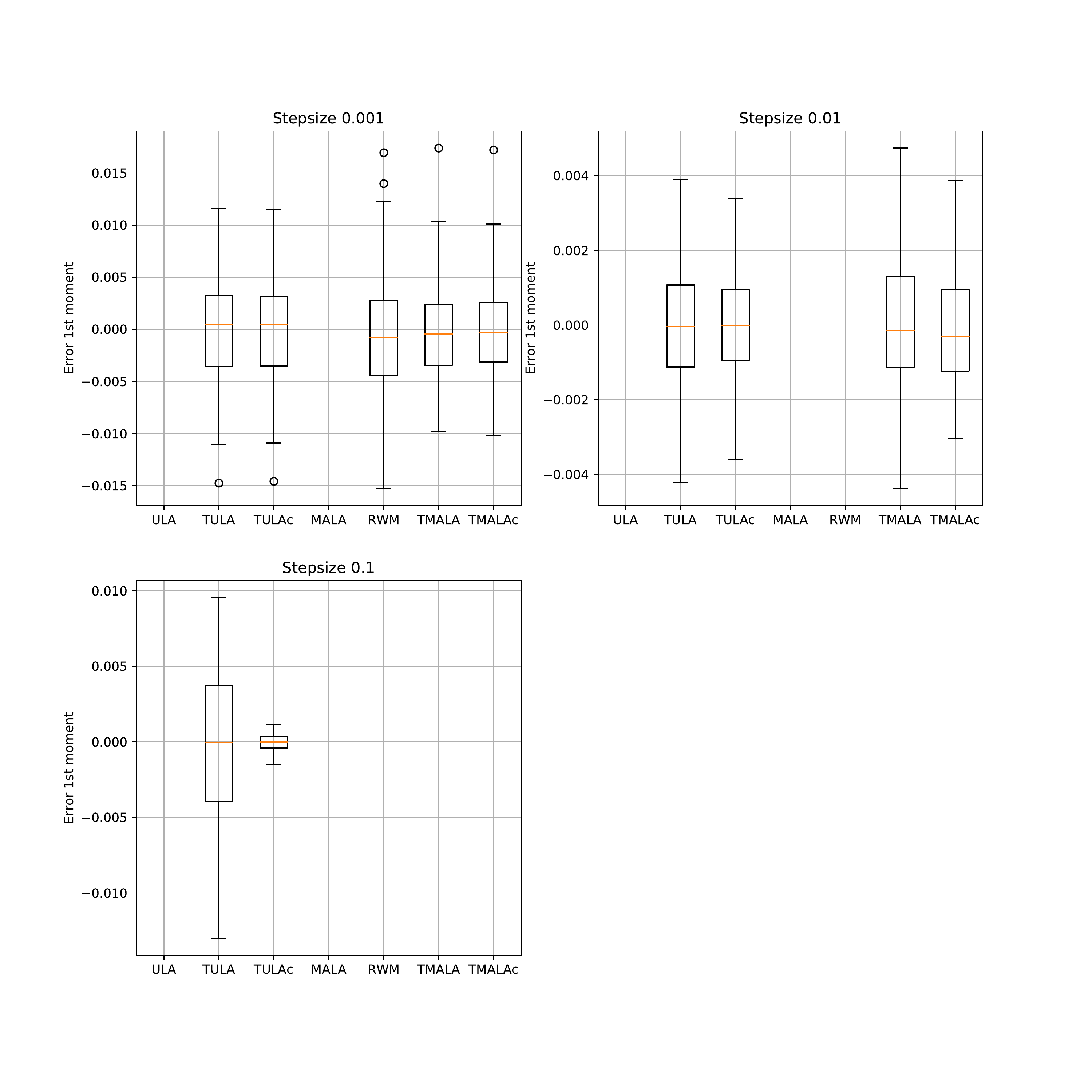}
%\end{center}
\caption{\label{figure:dw-1st-d100-x100} Boxplots of the error on the first moment for the double well in dimension $100$ starting at $(100, 0^{\otimes 99})$ for different step sizes.}
\end{figure}

\begin{figure}
%\begin{center}
\hspace*{-2cm}\includegraphics[scale=0.55]{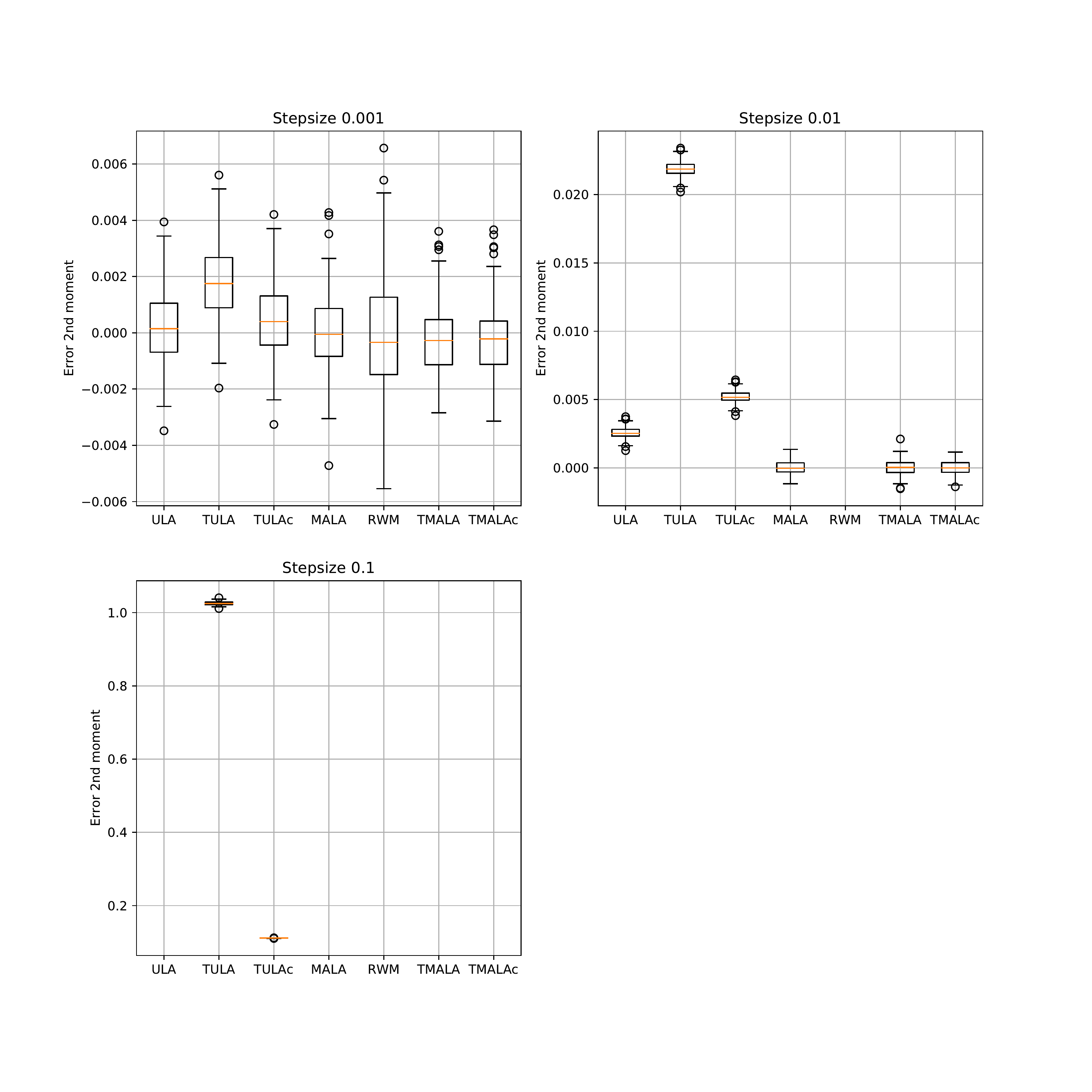}
%\end{center}
\caption{\label{figure:dw-2nd-d100-x0} Boxplots of the error on the second moment for the double well in dimension $100$ starting at $0$ for different step sizes.}
\end{figure}

\begin{figure}
%\begin{center}
\hspace*{-2cm}\includegraphics[scale=0.55]{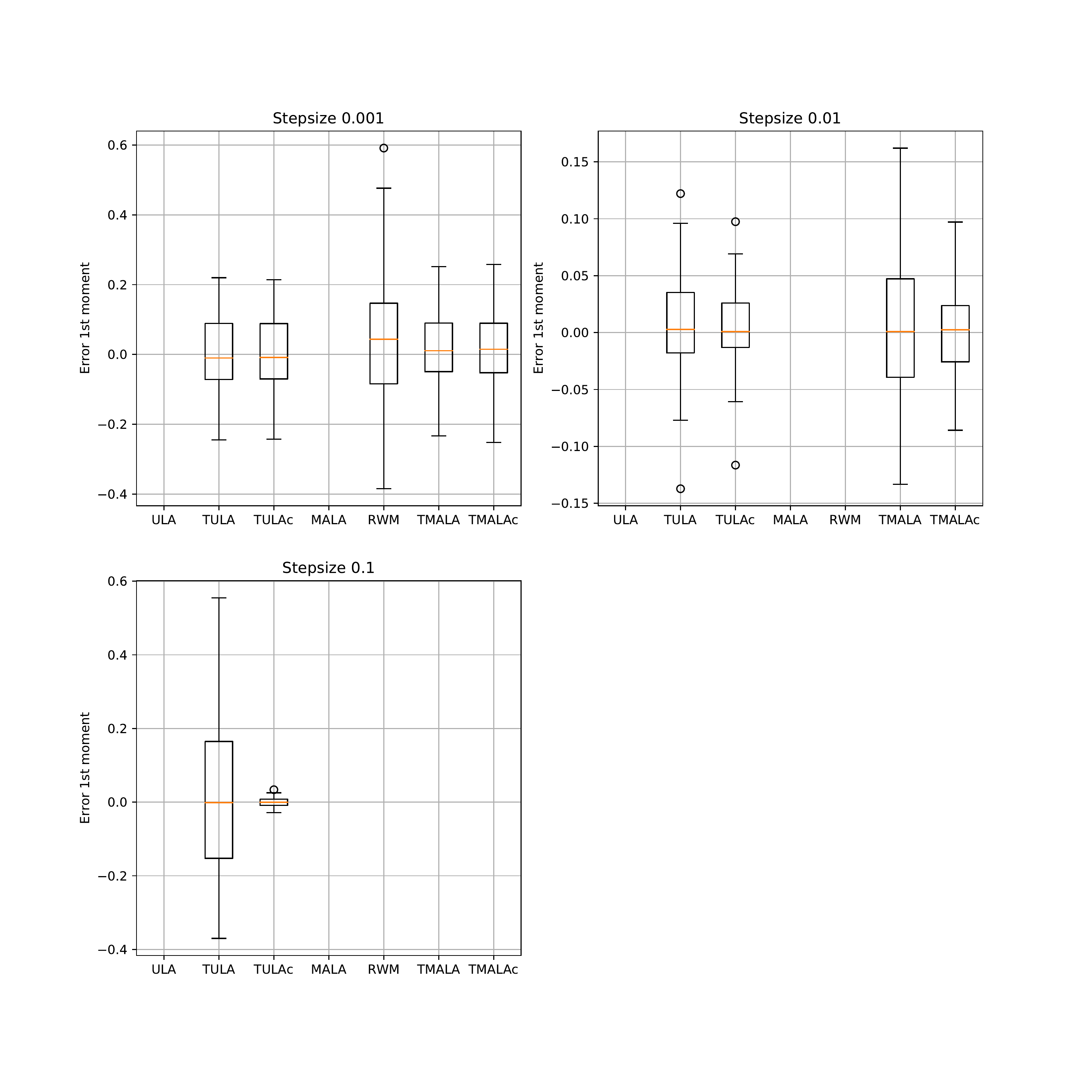}
%\end{center}
\caption{\label{figure:g-1st-d1000-x100} Boxplots of the error on the first moment for the Ginzburg-Landau model in dimension $1000$ starting at $(100, 0^{\otimes 999})$ for different step sizes.}
\end{figure}

%\begin{figure}
%%\begin{center}
%\hspace*{-2cm}\includegraphics[scale=0.55]{dw-1st-d100-x10-smart-legendfree.pdf}
%%\end{center}
%\caption{\label{figure:dw-1st-d100-x10-smart} Boxplots of the error on the first moment for the double well in dimension $100$ starting at $(10, 0^{\otimes 99})$ for different step sizes, including the results of the "smart" TULA.}
%\end{figure}

%%% Local Variables:
%%% mode: latex
%%% TeX-master: "../main"
%%% End:

%% file: proof.tex
\section{Proofs}
\label{sec:proofs}

\subsection{Proof of \Cref{prop:existence_ergodicity}}
\label{sec:proof-Lyapunov-langevin}

%\begin{proof}[Proof of \Cref{prop:existence_ergodicity}]
We have for all $x\in\rset^d$,
\begin{multline}
\label{eq:generator-Va}
\frac{\generator \lyapa(x)}{\a\lyapa(x)} = -\ps{\nablaU(x)}{\frac{x}{(1+\norm[2]{x})^{1/2}}} + \frac{\a\norm[2]{x}}{1+\norm[2]{x}} + \frac{d}{(1+\norm[2]{x})^{1/2}} \\
- \frac{\norm[2]{x}}{(1+\norm[2]{x})^{3/2}} \eqsp.
\end{multline}
By \Cref{assum:ergodicity-U}-\ref{item:ergodicity-U-2} and using $s\mapsto s/(1+s^2)^{1/2}$ is non-decreasing for $s\geq 0$, there exist $\MM_1, \kappaU \in\rset^*_+$ such that for all $x\in\rset^d$, $\norm{x} \geq \MM_1$, $\ps{\nablaU(x)}{x(1+\norm[2]{x})^{-1/2}} \geq \kappaU \norm{\nablaU(x)}$. By \Cref{assum:ergodicity-U}-\ref{item:ergodicity-U-1}, there exists $\MM_2 \geq \MM_1$ such that for all $x\in\rset^d$, $\norm{x} \geq \MM_2$, $\norm{\nablaU(x)} \geq \kappaU^{-1}\{1+\a+d(1+\MM_1^2)^{-1/2}\}$. We then have for all $x\in\rset^d$, $\norm{x} \geq \MM_2$, $\generator \lyapa(x) \leq - \a\lyapa(x)$. Define
\[ \bL[\a] = \exp(\a(1+\MM_2^2)^{1/2})\{2\LG(1+\MM_2^{\ell+1})+\a+d\} \eqsp.\]
Combining \eqref{eq:estimate_nabla_u_x} and \eqref{eq:generator-Va} gives \eqref{eq:drift-langevin}.
%Denote by $v_a(t,x) = \sgP_t \lyapa(x)$ for $x\in\rset^d$ and $t\geq 0$. We have
%\begin{equation*}
%(\partial / \partial t) v_a(t,x) \leq -\a v_a(t,x) + \a\bL
%\end{equation*}
%Applying \cite[Theorem 1.1]{meyn:tweedie:1993:III} with $V(x,t)=\lyapa(x)\rme^{\a t}$, $g_{-}(t) = 0$ and $g_{+}(t) = \a \bL[\a] \rme^{\a t}$ for $x\in\rset^d$ and $t\geq 0$, we get $ \sgP_t \lyapa(x) \leq \rme^{-\a t} \lyapa(x) + \bL[\a](1-\rme^{-\a t})$.
By \cite[Theorem 1.1]{meyn:tweedie:1993:III}, we get $ \sgP_t \lyapa(x) \leq \rme^{-\a t} \lyapa(x) + \bL[\a](1-\rme^{-\a t})$.
The second statement is a consequence of \cite[Theorem 2.2]{roberts:tweedie:1996} and \cite[Theorem 6.1]{meyn:tweedie:1993:III}.
%\end{proof}

\subsection{Proof of \Cref{lemma:checkA1A2}}
\label{sec:proof-lemma-checkA1A2}

%\begin{proof}[Proof of \Cref{lemma:checkA1A2}]
Let $\gastep>0$. We have for all $x\in\rset^d$, $\norm{\TU(x) - \nablaU(x)} \leq \gastep \norm[2]{\nablaU(x)}$ and
\[ \norm{\TUc(x) - \nablaU(x)} \leq \gastep \defEns{\sum_{i=1}^d \parenthese{\partial_i \pU(x)}^4}^{1/2} \leq \gastep \norm[2]{\nablaU(x)} \eqsp. \]
By \eqref{eq:estimate_nabla_u_x}, \Cref{assum:dr-close-nablaU} is satisfied.
Define for all $x\in\rset^d$, $x\neq 0$,
\begin{equation*}
A_\gastep(x) = \ps{\frac{x}{\norm{x}}}{\TU(x)} - \frac{\gastep}{2\norm{x}}\norm[2]{\TU(x)} \eqsp.
\end{equation*}
By \Cref{assum:ergodicity-U}-\ref{item:ergodicity-U-2}, there exist $M_1, \kappaU >0$ such that for all $x\in\rset^d$, $\norm{x} \geq M_1$, $\ps{x}{\nablaU(x)} \geq \kappaU\norm{x}\norm{\nablaU(x)}$. We get then for all $x\in\rset^d$, $\norm{x}\geq M_1$,
\begin{align*}
A_\gastep(x) &= \frac{1}{2\norm{x}\{1+\gastep\norm{\nablaU(x)}\}}\defEns{2\ps{x}{\nablaU(x)} - \norm{\nablaU(x)} \frac{\gastep\norm{\nablaU(x)}}{1+\gastep\norm{\nablaU(x)}}} \\
&\geq \frac{\norm{\nablaU(x)}}{1+\gastep\norm{\nablaU(x)}}\frac{2\kappaU\norm{x} -1}{2\norm{x}} \eqsp.
\end{align*}
By \Cref{assum:ergodicity-U}-\ref{item:ergodicity-U-1}, there exist $M_2, C>0$ such that for all $x\in\rset^d$, $\norm{x}\geq M_2$, $\norm{\nablaU(x)} \geq C$. Using that $s\mapsto s(1+\gastep s)^{-1}$ is non-decreasing for $s\geq 0$, we get for all $x\in\rset^d$, $\norm{x} \geq \max(\kappaU^{-1}, M_1, M_2)$, $A_\gastep(x) \geq (\kappaU C)/\{2(1+\gastep C)\}$.

Define for all $x\in\rset^d$, $x\neq 0$,
\begin{equation*}
B_\gastep(x) = \ps{\frac{x}{\norm{x}}}{\TUc(x)} - \frac{\gastep}{2\norm{x}}\norm[2]{\TUc(x)} \eqsp.
\end{equation*}
We have for all $x\in\rset^d$, $\gastep\norm{\TUc(x)} \leq \sqrt{d}$ and for all $x\in\rset^d$, $\norm{x} \geq M_1$,
\begin{equation*}
\ps{x}{\parenthese{\frac{\partial_i \pU(x)}{1+\gastep\absolute{\partial_i \pU(x)}}}_{i\in\defEns{1,\ldots,d}}}
\geq \frac{\kappaU\norm{x}\norm{\nablaU(x)}}{1+\gastep\max_{i\in\defEns{1,\ldots,d}} \absolute{\partial_i \pU(x)}}
\end{equation*}
and
\begin{equation*}
\norm{\parenthese{\frac{\partial_i \pU(x)}{1+\gastep\absolute{\partial_i \pU(x)}}}_{i\in\defEns{1,\ldots,d}}}
\leq \frac{\norm{\nablaU(x)}}{1+\gastep\max_{i\in\defEns{1,\ldots,d}} \absolute{\partial_i \pU(x)}} \eqsp.
\end{equation*}
Combining these inequalities, we get for all $x\in\rset^d$, $\norm{x} \geq \max(\kappaU^{-1} \sqrt{d}, M_1)$,
\begin{equation*}
B_\gastep(x) \geq \frac{\norm{\nablaU(x)}}{1+\gastep\max_{i\in\defEns{1,\ldots,d}} \absolute{\partial_i \pU(x)}} \frac{1}{2\norm{x}} \defEns{2\kappaU\norm{x} -\sqrt{d}}
\geq \frac{\norm{\nablaU(x)}}{1+\gastep\norm{\nablaU(x)}} \frac{\kappaU}{2} \eqsp,
\end{equation*}
and for all $x\in\rset^d$, $\norm{x} \geq \max(\kappaU^{-1} \sqrt{d}, M_1, M_2)$, we get $B_\gastep(x) \geq (\kappaU C)/\{2(1+\gastep C)\}$.
%\end{proof}

\subsection{Proof of \Cref{propo:drift_tamed_euler}}
\label{sec:proof-Lyapunov-discrete-Vae}

%Note that under \Cref{assum:ergodicity-U}, $\liminf_{\norm{x}\to\plusinfty} \pU(x) = \plusinfty$, $\pU$ has a minimum $\xstar$ and $\nablaU (\xstar) = 0$. Without loss of generality, it is assumed that $\xstar = 0$. It implies under \Cref{assum:grad-loc-lipschitz} that for all $x\in\rset^d$,
%\begin{equation}
%\label{eq:estimate_nabla_u_x}
%  \norm{\nabla U(x)} \leq 2 \LG \defEns{1+\norm[\ell+1]{x}}\eqsp.
%\end{equation}
%Besides, under \Cref{assum:ergodicity-U}-\ref{item:ergodicity-U-2}, there exists $C\in\rset$ such that for all $x\in\rset^d$, $  \ps{-\nabla U(x)}{x} \leq C$. By \cite[Theorem 2.1]{meyn:tweedie:1993:III}, \cite[Chapter IV, Theorems 2.3, 3.1]{ikeda:watanabe:1989} and \cite[Theorem 2.1]{roberts:tweedie:1996}, \eqref{eq:langevin} has a unique strong solution. The distribution of $(\XL_t)_{t\geq 0}$ defines a strongly Markovian semigroup $(\sgP_t)_{t\geq 0}$ given for all $t\geq 0$, $x\in\rset^d$ and $\eventA \in \Borel(\rset^d)$ by $\sgP_t(x, \eventA) = \expeMarkov{x}{\1\defEns{\XL_t \in \eventA}}$. $(\sgP_t)_{t\geq 0}$ is reversible \wrt~$\invpi$ and for all $x\in\rset^d$, $\lim_{t\to\plusinfty}\tvnorm{\sgP_t(x,\cdot) - \invpi} = 0$.
%Under \Cref{assum:grad-loc-lipschitz}, \Cref{assum:dr-close-nablaU} and by \eqref{eq:estimate_nabla_u_x}, we have for all $x\in\rset^d$
%\begin{equation}
%\label{eq:estimate-drift-step}
%\norm{\dr(x)} \leq 2 \LG \defEns{1+\norm[\ell+1]{x}} + \gastep \Cal \parenthese{1+\norm[\alphag]{x}} \eqsp.
%\end{equation}
%\begin{proof}[Proof of \Cref{propo:drift_tamed_euler}]
Let $\gastep, \a \in \rset_+^*$.
  Note that the function $x \mapsto (1+\norm[2]{x})^{1/2}$ is
  Lipschitz continuous with Lipschitz constant equal to $1$.  By the
  log-Sobolev inequality \cite[Proposition 5.5.1]{bakry:gentil:ledoux:2014}, and the Cauchy-Schwarz inequality, we have for all $x\in\rset^d$ and $a>0$
  \begin{align}
\nonumber
\rker \lyapa(x) &\leq \rme^{\a^2 \gastep}\exp\defEns{ \a \int_{\rset^d}(1+\norm[2]{y})^{1/2}\rker(x,\rmd y)}\\
    \label{eq:log_sob_0}
& \leq \rme^{\a^2 \gastep}\exp\defEns{\a \parenthese{1+\norm[2]{x- \gaStep \Ugad(x)} + 2 \gastep d}^{1/2}} \eqsp.
  \end{align}
We now bound the term inside the exponential in the right hand side. For all $x\in\rset^d$,
\begin{equation}
\label{eq:drift_discr_1}
\norm[2]{x- \gaStep \Ugad(x)}
=\norm[2]{x} -2 \gastep \parenthese{\ps{\dr(x)}{x} - (\gastep/2)\norm[2]{\dr(x)}} \eqsp.
\end{equation}
By \Cref{assum:ergodicity-drift}, there exist $\MM_1, \kappaU \in\rset^*_+$ such that for all $x\in\rset^d$, $\norm{x}\geq \MM_1$, $\ps{x}{\dr(x)} - (\gastep/2)\norm[2]{\dr(x)} \geq \kappaU \norm{x}$. Denote by $\MM = \max(\MM_1, 2d \kappaU^{-1})$. For all $x\in\rset^d$, $\norm{x}\geq \MM$, we have
\begin{equation*}
\norm[2]{x- \gaStep \Ugad(x)} + 2 \gastep d
  \leq \norm[2]{x} -\gastep \kappaU \norm{x} \eqsp.
\end{equation*}
Using for all $t\in \ccint{0,1}$, $(1-t)^{1/2} \leq 1-t/2$ and $s\mapsto s/(1+s^2)^{1/2}$ is non-decreasing for $s\geq 0$, we have for all $x\in\rset^d$, $\norm{x}\geq \MM$,
\begin{align*}
\parenthese{1+\norm[2]{x- \gaStep \Ugad(x)} + 2 \gastep d}^{1/2} &\leq
\parenthese{1+ \norm[2]{x}}^{1/2}
\parenthese{1- \frac{\gaStep \kappaU\norm{x}}{1+\norm[2]{x}}}^{1/2} \\
&\leq \parenthese{1+ \norm[2]{x}}^{1/2}
-\frac{\gaStep \kappaU\MM}{2(1+\MM^2)^{1/2}} \eqsp.
\end{align*}
Plugging this result in \eqref{eq:log_sob_0} shows that for all $x\in\rset^d$, $\norm{x} \geq \MM$,
\begin{equation}
  \label{eq:lem}
  \rker \lyape(x) \leq \rme^{-\ae^2 \gastep} \lyape(x) \quad \text{for} \quad
  \ae = \frac{\kappaU\MM}{4(1+\MM^2)^{1/2}} \eqsp.
\end{equation}
%\begin{equation*}
%\ce = \ae^2 + \ae
%\defEns{\MM \max_{\norm{x}\leq \MM} \norm{\dr(x)} +
%\frac{\gastep}{2}\parenthese{\max_{\norm{x}\leq \MM} \norm{\dr(x)}}^2
%+d} \eqsp.
%\end{equation*}
By \eqref{eq:estimate-drift-step}, we have
\begin{equation*}
\max_{\norm{x}\leq \MM} \norm{\dr(x)} \leq 2 \LG \defEns{1+\norm[\ell+1]{\MM}} + \gastep \Cal \parenthese{1+\norm[\alphag]{\MM}} \eqsp.
\end{equation*}
Combining it with \eqref{eq:log_sob_0}, \eqref{eq:drift_discr_1}, $s\mapsto s/(1+s^2)^{1/2}$ is non-decreasing for $s\geq 0$ and $(1+t_1+t_2)^{1/2} \leq (1+t_1)^{1/2} + t_2/2$ for $t_1=\norm[2]{x}$, $t_2=\gastep^2\norm[2]{\dr(x)} + 2\gastep\norm{x}\norm{\dr(x)} + 2\gastep d$, we have for all $x\in\rset^d$, $\norm{x} \leq \MM$,
\begin{equation}
  \label{eq:lem_2}
  \rker \lyape(x) \leq \rme^{ \gastep \ce  } \lyape(x) \eqsp,
\end{equation}
where
\begin{align*}
\ce &= \ae^2 + \ae
\bigg[
\MM \defEns{2 \LG \defEns{1+\norm[\ell+1]{\MM}} + \gastep \Cal \parenthese{1+\norm[\alphag]{\MM}}} \\
&\phantom{--------}+\frac{\gastep}{2}
\defEns{2 \LG \defEns{1+\norm[\ell+1]{\MM}} + \gastep \Cal \parenthese{1+\norm[\alphag]{\MM}}}^2
+d
\bigg] \eqsp.
\end{align*}
Then, using that for all $t \geq 0$, $1-\rme^{-t} \leq t$, we get for all $x\in\rset^d$, $\norm{x} \leq \MM$,
\begin{equation}
  \label{eq:lem_3}
  \rker \lyape(x)-\rme^{-\ae^2 \gastep} \lyape(x) \leq \rme^{ \gastep \ce  }(1-\rme^{-\gastep(\ae^2 + \ce)}) \lyape(x) \leq \gastep \rme^{ \gastep \ce  }(\ae^2 + \ce)\lyape(x)  \eqsp,
\end{equation}
which combined with \eqref{eq:lem} gives \eqref{eq:drift-rker} with $\be = \rme^{ \gastep \ce  }(\ae^2 + \ce) \rme^{\kappa M / 4}$.
Finally, using Jensen's inequality and $(s+t)^\varsigma \leq s^\varsigma + t^\varsigma$ for $\varsigma\in\ooint{0,1}$, $s,t\geq 0$ in \eqref{eq:drift-rker},
by \cite[Section 3.1]{roberts:tweedie:1996}, for all $\gastep>0$, $\rker$ has a unique invariant probability measure $\invpig$ and $\rker$ is $\lyape^\varsigma$-geometrically ergodic \wrt~$\invpig$.
%\end{proof}

\subsection{Proof of \Cref{prop:convergence-Vnorm}}
\label{sec:proof-convergence-Vnorm}

The proof is adapted from \cite[Proposition 2]{dalalyan:tsybakov:2012} and \cite[Theorem 10]{durmus2017}.
We first state a lemma.
\begin{lemma}
\label{lemma:kullback-leibler}
Assume \Cref{assum:grad-loc-lipschitz}, \Cref{assum:ergodicity-U}, \Cref{assum:dr-close-nablaU} and \Cref{assum:ergodicity-drift}. Let $\gastep_0>0$, $p\in\nset^*$ and $\nu_0$ be a probability measure on $(\rset^d, \Borel(\rset^d))$. There exists $C>0$ such that for all $\gastep\in\ocint{0,\gastep_0}$
\begin{equation*}
\KL(\nu_0 \rker^p | \nu_0 \sgP_{p\gastep}) \leq C\gastep^2 \int_{\rset^d} \sum_{i=0}^{p-1}
\defEns{\int_{\rset^d} \lyape(z) \rker^i(y, \rmd z)}
\nu_0(\rmd y) \eqsp.
\end{equation*}
%\begin{equation*}
%\KL(\nu_0 \rker^p | \nu_0 \sgP_{p\gastep}) \leq (\gastep^2/4) \int_{\rset^d} \sum_{i=0}^{p-1}
%\defEns{\int_{\rset^d} \defEns{\polyga[1](\norm{z}) + 2\gastep \polyp[2](\norm{z})} \rker^i(y, \rmd z)}
%\nu_0(\rmd y) \eqsp.
%\end{equation*}
\end{lemma}

\begin{proof}
Let $y\in\rset^d$ and $\gastep>0$.
%Denote by $\mu^y_p$ and $\mubar^y_p$ the laws on $\mathcal{C}(\ccint{0,p\gastep}, \rset^d)$ of the Langevin diffusion $(Y_t)_{t\in\ccint{0,p\gastep}}$ defined by \eqref{eq:langevin} and of the linear interpolation $(\Ybar_t)_{t\in\ccint{0,p\gastep}}$ of $(\XE_k)_{k\in\nset}$ defined by \eqref{eq:def_tamed_euler_1} both started at $y$.
Denote by $(Y_t, \Ybar_t)_{t\geq 0}$ the unique strong solution of
\begin{equation}
\label{eq:def_Y_Ybar}
  \begin{cases}
    & \rmd \XL_t = - \nabla U(\XL_t) \rmd t + \sqrt{2} \rmd B_t \quad, \quad \XL_0 = y \eqsp, \\
    & \rmd \Ybar_t =  - \dr\parenthese{\Ybar_{\floor{t/\gastep}\gastep}} \rmd t + \sqrt{2}\rmd B_t \quad, \quad  \Ybar_0 = y \eqsp,
  \end{cases}
\end{equation}
%where $\nablaUbar:\Csetfunction(\ccint{0,p\gastep},\rset^d) \times \ccint{0,p\gastep} \to \rset^d$ is defined for all $w\in\Csetfunction(\ccint{0,p\gastep},\rset^d)$ and $t\in\ccint{0,p\gastep}$ by
%\[ \nablaUbar(w, t) = \sum_{k=0}^{p-1} \dr(w_{k\gastep}) \1\defEns{\ccint{k\gastep, (k+1)\gastep}}(t) \eqsp.  \]
and by $(\filtration_t)_{t \geq 0}$ the filtration
associated with $(B_t)_{t \geq 0}$.
Denote by $\mu^y_p$ and $\mubar^y_p$ the marginal distributions on $\mathcal{C}(\ccint{0,p\gastep}, \rset^d)$ of $(Y_t, \Ybar_t)_{t\geq 0}$.
%The following theorem states that $\mu^y_p$ and $\mubar^y_p$ are equivalent and gives the associated Rakon-Nikodym derivative.
%\begin{theorem}[{{\cite[Theorem 7.19]{liptser2013statistics}}}]
%\label{thm:liptser}
%Let $(\XL_t, \Ybar_t)_{t\in\ccint{0,p\gastep}}$ be the unique strong solution of \eqref{eq:def_Y_Ybar}. If
%\begin{align*}
%\PP\parenthese{\int_0^{p\gastep} \norm[2]{\nablaU(\XL_t)} + \norm[2]{\dr\parenthese{\XL_{\floor{t/\gastep}\gastep}}} \rmd t < \plusinfty} &= 1 \eqsp,\\
%\PP\parenthese{\int_0^{p\gastep} \norm[2]{\nablaU(\Ybar_t)} + \norm[2]{\dr\parenthese{\Ybar_{\floor{t/\gastep}\gastep}}} \rmd t < \plusinfty} &= 1 \eqsp,
%\end{align*}
%then $\mu^y_p$ and $\mubar^y_p$ are equivalent and $\PP$-almost surely,
%%$\mubar^y_p$-almost surely,
%\begin{multline*}
%\frac{\rmd \mu^y_p}{\rmd \mubar^y_p}((\Ybar_t)_{t\in\ccint{0,p\gastep}}) = \exp
%\bigg(
%\frac{1}{2}\int_0^{p\gastep} \ps{-\nablaU(\Ybar_s) + \dr\parenthese{\Ybar_{\floor{s/\gastep}\gastep}}}{\rmd \Ybar_s} \\
%-\frac{1}{4}\int_0^{p\gastep} \defEns{\norm[2]{\nablaU(\Ybar_s)} - \norm[2]{\dr\parenthese{\Ybar_{\floor{s/\gastep}\gastep}}}} \rmd s
%\bigg)\eqsp.
%\end{multline*}
%\end{theorem}
%
%By \eqref{eq:estimate_nabla_u_x}, \eqref{eq:estimate-drift-step} and \Cref{propo:drift_tamed_euler,prop:existence_ergodicity}, the assumptions of \Cref{thm:liptser} are satisfied and we get
By \eqref{eq:estimate_nabla_u_x}, \eqref{eq:estimate-drift-step} and \Cref{propo:drift_tamed_euler,prop:existence_ergodicity}, we have
\begin{align*}
\PP\parenthese{\int_0^{p\gastep} \norm[2]{\nablaU(\XL_t)} + \norm[2]{\dr\parenthese{\XL_{\floor{t/\gastep}\gastep}}} \rmd t < \plusinfty} &= 1 \eqsp,\\
\PP\parenthese{\int_0^{p\gastep} \norm[2]{\nablaU(\Ybar_t)} + \norm[2]{\dr\parenthese{\Ybar_{\floor{t/\gastep}\gastep}}} \rmd t < \plusinfty} &= 1 \eqsp.
\end{align*}
By \cite[Theorem 7.19]{liptser2013statistics}, $\mu^y_p$ and $\mubar^y_p$ are equivalent and $\PP$-almost surely,
%$\mubar^y_p$-almost surely,
\begin{multline*}
\frac{\rmd \mu^y_p}{\rmd \mubar^y_p}((\Ybar_t)_{t\in\ccint{0,p\gastep}}) = \exp
\bigg(
\frac{1}{2}\int_0^{p\gastep} \ps{-\nablaU(\Ybar_s) + \dr\parenthese{\Ybar_{\floor{s/\gastep}\gastep}}}{\rmd \Ybar_s} \\
-\frac{1}{4}\int_0^{p\gastep} \defEns{\norm[2]{\nablaU(\Ybar_s)} - \norm[2]{\dr\parenthese{\Ybar_{\floor{s/\gastep}\gastep}}}} \rmd s
\bigg)\eqsp.
\end{multline*}
We get then
\begin{align*}
\KL(\mubar^y_p | \mu^y_p) &= \expe{-\log\defEns{\frac{\rmd \mu^y_p}{\rmd \mubar^y_p}((\Ybar_t)_{t\in\ccint{0,p\gastep}})}} \\
&= (1/4) \int_0^{p\gastep} \expe{\norm[2]{\nablaU(\Ybar_s)-\dr\parenthese{\Ybar_{\floor{s/\gastep}\gastep}}}} \rmd s \\
&= (1/4) \sum_{i=0}^{p-1} \int_{i\gastep}^{(i+1)\gastep} \expe{\norm[2]{\nablaU(\Ybar_s)- \Ugad(\Ybar_{i\gastep})}} \rmd s \eqsp.
\end{align*}
For $i\in\defEns{0,\ldots,p-1}$ and $s\in\coint{i\gastep,(i+1)\gastep}$, we have $\norm[2]{\nablaU(\Ybar_s)- \Ugad(\Ybar_{i\gastep})} \leq 2(A_1 + A_2)$ where
\begin{equation*}
A_1 = \norm[2]{\nablaU(\Ybar_s)-\nablaU(\Ybar_{i\gastep})} \eqsp, \quad
A_2 = \norm[2]{\nablaU(\Ybar_{i\gastep})- \Ugad(\Ybar_{i\gastep})} \eqsp.
\end{equation*}
By \Cref{assum:dr-close-nablaU}, $A_2 \leq \gastep^2 \Cal^2 \parenthese{1+\norm[\alphag]{\Ybar_{i\gastep}}}^2$ and by \Cref{assum:grad-loc-lipschitz},
\begin{equation}
\label{eq:temp-A1}
A_1 \leq \LG^2 \parenthese{1+\norm[\ell]{\Ybar_s}  + \norm[\ell]{\Ybar_{i\gastep}}}^2 \norm[2]{\Ybar_s - \Ybar_{i\gastep}} \eqsp.
\end{equation}
On the other hand for $s\in\coint{i\gastep,(i+1)\gastep}$,
\begin{align}
\nonumber
\norm[2]{\Ybar_s - \Ybar_{i\gastep}} &= (s-i\gastep)^2 \norm[2]{ \Ugad(\Ybar_{i\gastep})} + 2\norm[2]{B_s - B_{i\gastep}} \\
\label{eq:temp-A11}
&\phantom{--------}- 2^{3/2}(s-i\gastep) \ps{B_s - B_{i\gastep}}{\Ugad(\Ybar_{i\gastep})} \eqsp,\\
\label{eq:temp-A12}
\norm{\Ybar_{s}} &\leq \norm{\Ybar_{i\gastep}} + \gastep\norm{\dr(\Ybar_{i\gastep})} + \sqrt{2}\norm{B_s - B_{i\gastep}} \eqsp.
\end{align}
%Define $\polyga[1],\polyp[2]\in\Cpoly(\rset_+,\rset_+)$ for $t\in\rset_+$ by
Define $\polyga[1]:\rset_+ \to \rset_+$ for all $t\in\rset_+$ by
\begin{align}
\nonumber
&\polyga[1](t) = (2\uppi)^{-d/2} \LG^2 \int_{\rset^d}  \left[
2\norm[2]{z} + \gastep \defEns{2\LG(1+t^{\ell+1}) + \gastep\Cal(1+t^\alpha)}^2
\right] \\
\label{eq:def-polyp1}
&\times \left[
1+ t^\ell + \defEns{t + \gastep\parenthese{2\LG(1+t^{\ell+1}) + \gastep\Cal(1+t^\alphag)} + \sqrt{2\gastep}\norm{z}}^{\ell}
\right]^2
 \rme^{-\norm[2]{z}/2} \rmd z \eqsp.
% \label{eq:def-polyp2}
%&\polyp[2](t) = \Cal^2\parenthese{1+t^\alphag}^2 \eqsp.
\end{align}
By \eqref{eq:estimate-drift-step}, \eqref{eq:temp-A1}, \eqref{eq:temp-A11} and \eqref{eq:temp-A12}, we have for $i\in\defEns{0,\ldots,p-1}$
\[ \int_{i\gastep}^{(i+1)\gastep}
\CPE{A_1}{\filtration_{i\gamma}} \rmd s \leq (\gastep^2/2) \polyga[1](\norm{\Ybar_{i\gastep}})
\]
and we get
\begin{equation*}
\int_{i\gastep}^{(i+1)\gastep}
\CPE{\norm[2]{\nablaU(\Ybar_s)- \Ugad(\Ybar_{i\gastep})} }{\filtration_{i\gamma}} \rmd s
\leq \gastep^2 \defEns{\polyga[1](\norm{\Ybar_{i\gastep}})+2\gastep \polyp[2](\norm{\Ybar_{i\gastep}})} \eqsp,
\end{equation*}
where $\polyp[2]:\rset_+ \to \rset_+$ is defined for all $t\in \rset_+$ by
\begin{equation}
\label{eq:def-polyp2}
\polyp[2](t) = \Cal^2\parenthese{1+t^\alphag}^2 \eqsp.
\end{equation}
By \cite[Theorem 4.1, Chapter 2]{kullback:1997}, we obtain
\begin{equation*}
\KL(\delta_y \rker^p | \delta_y \sgP_{p\gastep}) \leq
\KL(\mubar^y_p | \mu^y_p) \leq (\gastep^2/4) \sum_{i=0}^{p-1} \expe{\polyga[1](\norm{\Ybar_{i\gastep}}) + 2\gastep\polyp[2](\norm{\Ybar_{i\gastep}})} \eqsp.
\end{equation*}
By \eqref{eq:def-polyp1} and \eqref{eq:def-polyp2}, there exists $C>0$ such that for all $\gastep\in\ocint{0,\gastep_0}$ and $x\in\rset^d$, $\polyga[1](\norm{x}) +2\gastep\polyp[2](\norm{x}) \leq 4C \lyape(x)$. Combining it with the chain rule for the Kullback-Leibler divergence concludes the proof.
\end{proof}

\begin{proof}[Proof of \Cref{prop:convergence-Vnorm}]
Let $\gastep\in\ocint{0,\gastep_0}$.
By \Cref{prop:existence_ergodicity}, we have for all $n\in\nset$ and $x\in\rset^d$,
\begin{equation*}
\VnormEq[\lyape^{1/2}]{\delta_x \rker^n - \invpi} \leq \Cae \rhoae^{n\gastep} \lyape^{1/2}(x) + \VnormEq[\lyape^{1/2}]{\delta_x \rker^n - \delta_x \sgP_{n\gastep}} \eqsp.
\end{equation*}
Denote by $\kga = \ceil{\gastep^{-1}}$ and by $\qga,\rga$ the quotient and the remainder of the Euclidian division of $n$ by $\kga$. We have $\Vnorm[\lyape^{1/2}]{\delta_x \rker^n - \delta_x \sgP_{n\gastep}} \leq A+B$ where
\begin{align}
\nonumber
A &=  \VnormEq[\lyape^{1/2}]{\delta_x \rker^{\qga\kga} \sgP_{\rga\gastep} - \delta_x \rker^n} \\
\nonumber
B &= \sum_{i=1}^{\qga} \VnormEq[\lyape^{1/2}]{\delta_x \rker^{(i-1)\kga} \sgP_{(n-(i-1)\kga) \gastep} - \delta_x \rker^{i\kga} \sgP_{(n-i\kga) \gastep}} \\
\label{eq:tempo-B1}
&\leq \sum_{i=1}^{\qga} \Cae \rhoae^{(n-i\kga)\gastep} \VnormEq[\lyape^{1/2}]{\delta_x \rker^{(i-1)\kga} \sgP_{\kga \gastep} - \delta_x \rker^{i\kga}} \eqsp.
\end{align}
For $i\in\defEns{1,\ldots,\qga}$ we have by \cite[Lemma 24]{durmus2017},
\begin{multline}
\label{eq:tempo-B2}
\VnormEq[\lyape^{1/2}]{\delta_x \rker^{(i-1)\kga} \sgP_{\kga \gastep} - \delta_x \rker^{i\kga}}^2
\leq 2 \defEns{\delta_x \rker^{(i-1)\kga} \sgP_{\kga \gastep}(\lyape) + \delta_x \rker^{i\kga}(\lyape)} \\
\times \KL(\delta_x \rker^{i\kga} | \delta_x \rker^{(i-1)\kga} \sgP_{\kga \gastep}) \eqsp.
\end{multline}
By \Cref{propo:drift_tamed_euler}, \Cref{lemma:kullback-leibler} and $\kga \leq 1 + \gastep^{-1}$, we have for all $i\in\defEns{1,\ldots,\qga}$
\begin{align}
\nonumber
\KL(\delta_x \rker^{i\kga} | \delta_x \rker^{(i-1)\kga} \sgP_{\kga \gastep})
&\leq C\gastep^2 \sum_{j=0}^{\kga-1} \int_{\rset^d} \lyape(z) \delta_x \rker^{(i-1)\kga + j}(\rmd z) \\
\label{eq:tempo-B3}
&\leq C\gastep^2(1+\gastep^{-1})
\defEns{\rme^{-\ae^2 \gastep \kga (i-1)}\lyape(x) + \frac{\be}{\ae^2}\rme^{\ae^2 \gastep}} \eqsp,
\end{align}
where $C$ is the constant defined in \Cref{lemma:kullback-leibler}.
By \Cref{prop:existence_ergodicity}, we have for $x\in\rset^d$, $\sgP_{\kga \gastep} \lyape(x) \leq \lyape(x) + \bL[\ae]$ and by \Cref{propo:drift_tamed_euler}, we get for all $i\in\defEns{1,\ldots,\qga}$
\begin{equation}
\label{eq:tempo-B4}
\delta_x \rker^{(i-1)\kga} \sgP_{\kga \gastep}(\lyape) + \delta_x \rker^{i\kga}(\lyape)
\leq 2\defEns{\rme^{-\ae^2 \gastep \kga (i-1)}\lyape(x) + \frac{\be}{\ae^2}\rme^{\ae^2 \gastep} + \bL[\ae]} \eqsp.
\end{equation}
By \eqref{eq:tempo-B1}, \eqref{eq:tempo-B2}, \eqref{eq:tempo-B3} and \eqref{eq:tempo-B4}, we obtain
\begin{multline*}
B \leq 2\Cae C^{1/2} \gastep(1+\gastep^{-1})^{1/2} \\ \times
\sum_{i=1}^{\qga} \rhoae^{(\qga-i)\gastep\kga}
\defEns{
\rme^{-(i-1) \gastep\kga \ae^2} \lyape(x) + \parenthese{\bL[\ae]+\frac{\be}{\ae^2}\rme^{\ae^2 \gastep}}
} \eqsp
\end{multline*}
and we get
%Define $\mn = \min(\rhoae, \rme^{-\ae^2})$ and $\mx = \max(\rhoae, \rme^{-\ae^2})$. We have
\begin{multline*}
B \defEns{2\Cae C^{1/2} \gastep(1+\gastep^{-1})^{1/2}}^{-1} \leq
\parenthese{\bL[\ae]+\frac{\be}{\ae^2}\rme^{\ae^2 \gastep}}
\frac{1}{1-\rhoae^{\kga\gastep}} \\
+ \lyape(x) \qga \max(\rhoae, \rme^{-\ae^2})^{(\qga-1)\gastep\kga} \eqsp.
%\frac{1}{1-(\mn/\mx)^{\kga\gastep}}
\end{multline*}
Bounding $A$ along the same lines and using $\kga\gastep \geq 1$, we get \eqref{eq:thm-conv-Vnorm-1}.
By \Cref{propo:drift_tamed_euler} and taking the limit $n\to\plusinfty$, we obtain \eqref{eq:thm-conv-Vnorm-2}.
\end{proof}

\subsection{Proofs of \Cref{thm:wasserstein-gamma2,thm:wasserstein-gamma3}}
\label{sec:proof-strongly-convex}

%We begin with giving a bound on the $2p$-th moment of $\delta_x \sgP_t$ and $\pi$ for $t\geq 0$ and $p \in\nset^*$.
We first state preliminary technical lemmas on the diffusion $(\XL_t)_{t\geq 0}$. The proofs are postponed to the Appendix.
%Following \cite[Section 6.6]{durmuthese2016},
Define for all $p \in\nset^*$ and $k \in \{0,\cdots,p\}$,
\begin{equation}
\label{eq:def_akp}
a_{k,p} =
 m^{k-p}\prod_{i=k+1}^p \defEns{i(d+2(i-1))(i-k)^{-1}}  \eqsp.
\end{equation}

\begin{lemma}
\label{lem:moment_diffusion_2_p}
Assume \Cref{assum:convex}. Let $p \in \nset^*$, $x \in \rset^d$ and $(\XL_t)_{t \geq 0}$ be the solution of
\eqref{eq:langevin} started at $x$.
 For all $t \geq 0$,
\begin{equation*}
\expe{\norm[ 2 p]{\XL_t}}  \leq  a_{0,p}\parenthese{1-\rme^{-2pmt}}  +\sum_{k=1}^p a_{k,p} \rme^{-2kmt}\norm[2k]{x}  \eqsp,
\end{equation*}
where  for $ k \in \defEns{0,\cdots,p}$, $a_{k,p}$ is given in \eqref{eq:def_akp}.
\end{lemma}

\begin{proof}
The proof is postponed to \Cref{appendix:proof-lem-moment}.
%  By \Cref{assum:convex}, \eqref{eq:langevin} has a unique strong
%  solution $(\XL_t)_{t\geq 0}$ for any initial data $Y_0=x \in \rset^d$.
%  Define for $p \in\nset^*$, $V_p:\rset^d \to \rset_+$ by $\VStar_p(y) =
%  \norm[2p]{y}$ for $y\in\rset^d$. We have using \Cref{assum:convex},
%\begin{align*}
%%\label{eq:drift-Vp}
%  \generator V_p (x) &= -2p \norm[2(p-1)]{x}\ps{\nabla U(x)}{x} + 2p(d + 2(p-1))  \norm[2(p-1)]{x}\\
%&\leq -2p m \norm[2p]{x} + 2p\norm[2(p-1)]{x}(d + 2(p-1)) \eqsp.
%\end{align*}
%Applying \cite[Theorem 1.1]{meyn:tweedie:1993:III} with $V(x,t)=V_p(x)\rme^{2pm t}$, $g_{-}(t) = 0$ and $g_{+}(x,t) = 2p(d+2(p-1)) V_{p-1}(x) \rme^{2pm t}$ for $x\in\rset^d$ and $t\geq 0$, we get denoting by $v_p(t,x) = P_t V_p(x)$,
%%By  \cite[Chapter VII, proposition 1.2]{revuz:yor:1999}, , we have
%%\begin{equation*}
%%  \frac{\partial v_p(t,x)}{\partial t} \leq - 2pm v_p(t,x) + 2p(d + 2(p-1)) v_{p-1}(t,x) \eqsp.
%%\end{equation*}
%%By Grönwall's inequality, we get for all $x \in \rset^d$ and $t \geq 0$,
%\begin{equation*}
%  v_p(t,x)
% \leq \rme^{-2pmt} V_p(x) +  2p(d + 2(p-1)) \int_0^t \rme^{-2pm(t-s)} v_{p-1}(s,x) \rmd s \eqsp.
%\end{equation*}
%A straightforward induction concludes the proof.
\end{proof}

\begin{lemma}
  \label{theo:wasser_p_diffusion}
Assume  \Cref{assum:convex} and let $p \in \nset^*$.
We have $\int_{\rset^d}  \norm[2p]{y} \pi(\rmd y)  \leq a_{0,p}$.
\end{lemma}

\begin{proof}
The proof is postponed to \Cref{appendix:proof-theo-wasser}.
%Under \Cref{assum:convex}, by \cite{bolley:gentil:guillin:2012}, \eqref{eq:langevin} has a unique reversible measure $\invpi$ and \\
%${\lim_{t\to\plusinfty} W_{2p}(\delta_0 \sgP_t, \invpi) = 0}$. \cite[Theorem 6.9]{VillaniTransport} and \Cref{lem:moment_diffusion_2_p} conclude the proof.
\end{proof}

%Define
%\begin{equation}\label{eq:def-N}
%\N = \ceil{(\ell+1)/2} \eqsp.
%\end{equation}
%and for $p\in\nset^*$ and $x\in\rset^d$ $\bp:\rset^d\to\rset_+$ by,
%\begin{equation}
%\label{eq:def_bxp}
%  \bp(x) = \prod_{i=2}^p 2\parenthese{(\LG^2/m)\defEns{1+\norm[\ell+1]{x}}^2+d+2(i-1)} \eqsp.
%\end{equation}

Let $\gastep>0$ and under \Cref{assum:grad-loc-lipschitz} set
\begin{equation}
\label{eq:def-N}
N = \ceil{(\ell+1)/2} \eqsp.
\end{equation}
Consider %$\polyga[3]\in\Cpoly(\rset_+,\rset_+)$
$\polyga[3]:\rset_+ \to \rset_+$ defined for all $s\in\rset_+$ by
\begin{equation}
\label{eq:def-polyga3}
\polyga[3](s) = 2d + 8\LG^2(1+s^{\ell+1})
\defEns{\frac{\gastep}{2} \parenthese{2+\sum_{k=1}^\N a_{k,\N} s^{2k}} + \N \mU a_{0,\N} \frac{\gastep^{2}}{3}} \eqsp.
\end{equation}

\begin{lemma}
\label{item:moment_diffusion_3_p}
Assume \Cref{assum:grad-loc-lipschitz} and \Cref{assum:convex}.
%Let $p \in \nset^*$,
Let $x \in \rset^d$, $\gastep>0$ and $(\XL_t)_{t \geq 0}$ be the solution of \eqref{eq:langevin} started at $x$. For all $t\in\ccint{0,\gastep}$, we have
  %\begin{equation*}
 $\expe{\norm[2]{\XL_t-x}} \leq t \polyga[3](\norm{x})$, where $\polyga[3]$ is defined in \eqref{eq:def-polyga3}.
 %\end{equation*}
% \begin{align*}
% &\bp^{-1}(x)\expeMarkov{x}{\norm[2p]{\XL_t-x}} \leq
% 2d t^{p} + 8\LG^2(1+\norm[\ell+1]{x}) \\
% &\phantom{-------}\times \defEns{\frac{t^{p+1}}{p+1} \parenthese{2+\sum_{k=1}^\N a_{k,\N} \norm[2k]{x}} + 2\N \mU a_{0,\N} \frac{t^{p+2}}{(p+1)(p+2)}}
% \end{align*}
%  \begin{equation*}
% \expeMarkov{x}{\norm[2]{\XL_t-x}} \leq t \polyga[3](\norm{x})
% \end{equation*}
%   \begin{equation*}
% \expeMarkov{x}{\norm[2]{\XL_t-x}} \leq
% 2d t + 8\LG^2(1+\norm[\ell+1]{x})
%\defEns{\frac{t^{2}}{2} \parenthese{2+\sum_{k=1}^\N a_{k,\N} \norm[2k]{x}} + \N \mU a_{0,\N} \frac{t^{3}}{3}}
% \end{equation*}
% where,
% \begin{equation}\label{eq:def-N}
%\N = \ceil{(\ell+1)/2} \eqsp.
%\end{equation}
\end{lemma}

\begin{proof}
The proof is postponed to \Cref{appendix:proof-item-moment}.
\end{proof}

For $p\in\nset$ and $\gastep>0$, define $\polygaq[p]:\rset_+ \to \rset_+$ for all $s\in\rset_+$ by,
\begin{align}
\nonumber
&\polygaq[p](s) =
\defEns{\prod_{i=1}^p 2i(d+3i-2)}
\Bigg[
2d\frac{\gastep^{p}}{(p+1)!} + 8\LG^2(1+s^{\ell+1}) \\
\nonumber
&\phantom{-------}\times
	\bigg\{
	(2+\sum_{k=1}^N a_{k,N} s^{2k}) \frac{\gastep^{p+1}}{(p+2)!} + 2N m a_{0,N} \frac{\gastep^{p+2}}{(p+3)!}
	\bigg\}
\Bigg] \\
\nonumber
&\phantom{---} + 2\sum_{k=1}^p \defEns{\prod_{i=k+1}^p 2i(d+3i-2)} \defEns{d+4+\frac{\LG^2(1+s^{\ell+1})^2}{m(k+1)}} \\
\label{eq:polygaq}
&\phantom{-------} \times
\defEns{\parenthese{\sum_{i=1}^k a_{i,k} s^{2i}} \frac{\gastep^{p-k}}{(p+1-k)!} + 2km a_{0,k} \frac{\gastep^{p+1-k}}{(p+2-k)!}}
\end{align}
where $N$ is defined in \eqref{eq:def-N}.

\begin{lemma}
\label{item:moment_diffusion_4_p}
Assume \Cref{assum:grad-loc-lipschitz} and \Cref{assum:convex}. Let $p \in \nset$, $\gastep>0$, $x \in \rset^d$ and $(\XL_t)_{t \geq 0}$ be the solution of \eqref{eq:langevin} started at $x$. For all $t \in\ccint{0,\gastep}$, we have
 %\begin{equation*}
$\expe{\norm[2p]{\XL_t}\norm[2]{\XL_t - x}} \leq
t\polygaq[p](\norm{x})$, where $\polygaq[p]$ is defined in \eqref{eq:polygaq}.
 %\end{equation*}
%where $N$ is defined in \eqref{eq:def-N}.
\end{lemma}

\begin{proof}
The proof is postponed to \Cref{appendix:proof-item-diffusion}.
\end{proof}

\begin{lemma}
\label{lem:reg-hessian-U-growth}
Assume \Cref{assum:hessian}.
\begin{enumerate}[label = \alph*)]
\item
\label{item:reg-hessian-U-growth-1}
For all $x \in \rset^d$, $\normLigne{\nablaS U(x) } \leq \CH \defEnsLigne{1+\normLigne[\nue+\betah]{x}}$ where $\CH = \max(2\LH,\norm{\nablaS U(0)})$.
\item
\label{item:reg-hessian-U-growth-2}
For all $x,y \in \rset^d$,
\begin{equation*}
\norm{\nabla U(x) - \nabla U(y)-\nablaS U(y)(x-y)} \leq \frac{2\LH}{1+\betah} \defEns{1+ \norm[\nue]{x} + \norm[\nue]{y}}\norm[1+\betah]{x-y} \eqsp.
\end{equation*}
\end{enumerate}
\end{lemma}

\begin{proof}
\begin{enumerate}[label=\alph*), wide=0pt, labelindent=\parindent]
\item
By \Cref{assum:hessian}, we get for all $x \in \rset^d$
\begin{align*}
\norm{\nabla^2 U(x)}& \leq  \norm{\nablaS U(x)-\nablaS U(0)} + \norm{\nablaS U(0)} \\
&\leq \LH \defEns{1+\norm[\nue]{x}} \norm[\betah]{x} + \norm{\nablaS U(0)}  \eqsp.
\end{align*}
The proof then follows from the upper bound for all $x \in \rset^d$, $\norm[\betah]{x} \leq 1+\norm[\nue + \betah]{x}$.
\item
Let $x,y \in \rset^d$. By \Cref{assum:hessian},
\begin{align*}
&\norm{\nabla U(x) - \nabla U(y)-\nablaS U(y)(x-y)} \\
&\phantom{-------}\leq
\int_0^1 \norm{\nablaS \pU (tx+(1-t)y) - \nablaS \pU(y)} \rmd t \norm{x-y} \\
&\phantom{-------}\leq \LH \int_0^1 \defEns{1+\norm[\nue]{y} + \norm[\nue]{tx+(1-t)y}} \norm[\betah]{t(x-y)} \rmd t \norm{x-y} \eqsp,
\end{align*}
and the proof follows from $\norm[\nue]{tx+(1-t)y} \leq \norm[\nue]{x} + \norm[\nue]{y}$.
\end{enumerate}
\end{proof}

%If $\pU$ is twice continuously differentiable and there exist $\ell,\LG \geq 0 $ such that for all $x \in \rset^d$,
%\begin{equation}
%\label{lemma:prac-grad-loc-lipschitz}
%\norm{\nabla^2 U(x)} \leq \LG \defEns{1+ \norm[\ell]{x}} \eqsp,
%\end{equation}
%then \Cref{assum:grad-loc-lipschitz} is satisfied.
%Note that \Cref{assum:hessian} implies \Cref{assum:grad-loc-lipschitz} by \Cref{lem:reg-hessian-U-growth}-\ref{item:reg-hessian-U-growth-1} and \eqref{lemma:prac-grad-loc-lipschitz} with $\LG=\CH$ and $\ell = \nu+\betah$.

For all $n \in \nset$, we now bound the Wasserstein distance $W_2$ between $\pi$ and
the distribution of the $n^{\text{th}}$ iterate of
%the Euler discretization
$\XE_n$ defined by \eqref{eq:def_tamed_euler_1}.
%Following \cite{2016arXiv160501559D},
The strategy consists given two initial conditions $(x,y)$, in coupling
$\XE_n$ and $\XL_{\gamma n}$ solution of \eqref{eq:langevin} at time $\gamma n$, using the same
Brownian motion.
Similarly to \eqref{eq:def_Y_Ybar}, for $\gastep>0$, consider the
unique strong solution $(\XL_t, \Ybar_t)_{t\geq 0}$ of
%processes $(\XL_{t})_{t \geq 0}$ and the linear interpolation of $(\XE_n)_{n \in\nset}$, $(\XEL_t)_{t\geq 0}$ started at $y \in \rset^d$ and $x\in \rset^d$ respectively and defined by
\begin{equation}
\label{eq:def_coupling}
  \begin{cases}
    & \rmd \XL_t = - \nabla U(\XL_t) \rmd t + \sqrt{2} \rmd B_t \quad, \quad \XL_0 = y \eqsp, \\
    & \rmd \XEL_{t} =  - \Ugad\parenthese{\XEL_{\floor{t/\gastep}\gastep}} \rmd t + \sqrt{2} \rmd B_{t} \quad , \quad \Ybar_0 = x \eqsp,
  \end{cases}
\end{equation}
%More formally, consider the processes $(\XL_{t})_{t \geq 0}$ and  $(\XE_n)_{n \geq 0}$ started at $y \in \rset^d$ and $x\in \rset^d$ respectively and defined by
%\begin{equation}
%\label{eq:def_coupling}
%  \begin{cases}
%    & \rmd \XL_t = - \nabla U(\XL_t) \rmd t + \sqrt{2} \rmd B_t \\
%    & \XE_{n+1} =  \XE_n - \gastep \nabla \Ugad(\XE_n) + \sqrt{2}\parenthese{B_{(n+1)\gamma} -  B_{n\gamma }} \text{ for all $n \in \nset$} \eqsp,
%  \end{cases}
%\end{equation}
where $(B_t)_{t \geq 0}$ is a $d$-dimensional Brownian motion.
Note that for $n\in\nset$, $\XEL_{n\gastep} = \XE_n$ and let $(\filtration_t)_{t \geq 0}$ be the filtration
associated with $(B_t)_{t \geq 0}$.

\begin{lemma}
\label{lem:wasserstein-contraction-gamma2}
Assume \Cref{assum:dr-close-nablaU}, \Cref{assum:ergodicity-drift},  \Cref{assum:grad-loc-lipschitz} and \Cref{assum:convex}. Let $\gastep_0>0$.
Define $(\XL_{t})_{t \geq 0}$, $(\XEL_t)_{t\geq 0}$ by \eqref{eq:def_coupling}.
%and $X_n = \XEL_{n\gastep}$ for $n\in\nset$.
Then there exists $C>0$ such that for all $n \in \nset$ and $\gastep\in\ocint{0,\gastep_0}$, almost surely,
\begin{equation*}
\CPE{\norm[2]{Y_{(n+1)\gamma}- \XEL_{(n+1)\gamma}}}{\filtration_{n\gamma}} \leq
\rme^{-\mU \gastep} \norm[2]{\XL_{n\gastep} - \XEL_{n\gamma}}
+ C \gastep^2 \lyape(\XEL_{n\gamma}) \eqsp.
\end{equation*}
\end{lemma}

\begin{proof}
Using the Markov property, we only need to show the result for $n=0$. Define for $t\in\coint{0, \gastep}$, $\Theta_t = \XL_t - \XEL_t$. By It\^o's formula, we have for all $t\in\coint{0, \gastep}$,
\begin{equation*}
\norm[2]{\Theta_t} = \norm[2]{y-x} - 2\int_{0}^t \ps{\Theta_s}{\nablaU(\XL_s) - \Ugad (x)} \rmd s \eqsp.
\end{equation*}
By \eqref{eq:estimate_nabla_u_x} and \Cref{lem:moment_diffusion_2_p}, the family of random variables $\parenthese{\ps{\Theta_s}{\nablaU(\XL_s) - \Ugad (x)}}_{s\in\coint{0,\gastep}}$ is uniformly integrable. Pathwise continuity implies then for $s\in\coint{0, \gastep}$ the continuity of $s \mapsto \expe{\ps{\Theta_s}{\nablaU(\XL_s) - \Ugad (x)}}$.
Taking the expectation and deriving, we have for $t\in\coint{0,\gastep}$,
\begin{align}
\nonumber
\frac{\rmd}{\rmd t} \expe{\norm[2]{\Theta_t}}
&= -2 \expe{\ps{\Theta_t}{\nablaU(\XL_t) - \Ugad (x)}} \\
\nonumber
&= -2 \expe{\ps{\Theta_t}{\nablaU(\XL_t) - \nablaU (\XEL_t)}}
-2 A_1 -2 A_2 \\
\label{eq:derivative-thetat}
&\leq -2\mU \expe{\norm[2]{\Theta_t}} -2 A_1 -2 A_2 \eqsp,
\end{align}
where
\begin{equation}
\label{eq:A1A2}
A_1 = \expe{\ps{\Theta_t}{\nablaU(\XEL_t) - \nablaU (x)}} \eqsp,\quad
A_2 = \expe{\ps{\Theta_t}{\nablaU(x) - \Ugad (x)}} \eqsp.
\end{equation}
Using that $\absolute{\ps{a}{b}} \leq (\mU/4)\norm[2]{a} + \mU^{-1} \norm[2]{b}$ for all $a,b\in\rset^d$,
\begin{equation*}
\absolute{A_1} \leq (m/4) \expe{\norm[2]{\Theta_t}} + m^{-1} \expe{\norm[2]{\nablaU(\XEL_t) - \nablaU (x)}} \eqsp.
\end{equation*}
Similarly to the proof of \Cref{lemma:kullback-leibler}, we have $\expe{\norm[2]{\nablaU(\XEL_t) - \nablaU (x)}} \leq t\polyga[1](\norm{x})$ where $\polyga[1]$ is defined in \eqref{eq:def-polyp1}.
For $A_2$, we have
\begin{equation}
\label{eq:A2}
\absolute{A_2} \leq (m/4) \expe{\norm[2]{\Theta_t}} + m^{-1} \norm[2]{\nablaU(x) - \nabla \Ugad (x)}
\end{equation}
and $\norm[2]{\nablaU(x) - \Ugad (x)} \leq \gastep^2 \polyp[2](\norm{x})$ where $\polyp[2]$ is defined in \eqref{eq:def-polyp2}.
We get for $t\in\coint{0,\gastep}$,
\begin{equation*}
\frac{\rmd}{\rmd t} \expe{\norm[2]{\Theta_t}} \leq -m \expe{\norm[2]{\Theta_t}} + 2m^{-1} \defEns{t \polyga[1](\norm{x}) + \gastep^2 \polyp[2](\norm{x})} \eqsp.
\end{equation*}
Using Gr\"onwall's lemma and $1-\rme^{-s} \leq s$ for all $s\geq 0$, we obtain
\begin{equation*}
\expe{\norm[2]{Y_{\gamma}- \XEL_{\gastep}}} \leq
\rme^{-\mU \gastep} \norm[2]{y - x}
+ \mU^{-1} \gastep^2 \defEns{\polyga[1](\norm{x}) + 2\gastep \polyp[2](\norm{x})} \eqsp.
\end{equation*}
Finally, by \eqref{eq:def-polyp1} and \eqref{eq:def-polyp2}, there exists $C>0$ such that for all $x\in\rset^d$, $\polyga[1](\norm{x}) + 2\gastep \polyp[2](\norm{x}) \leq C \mU \lyape(x)$.
\end{proof}

\begin{lemma}
\label{lem:wasserstein-contraction-gamma3}
Assume \Cref{assum:dr-close-nablaU}, \Cref{assum:ergodicity-drift}, \Cref{assum:convex} and \Cref{assum:hessian}.
Let $\gastep_0>0$.
Define $(\XL_{t})_{t \geq 0}$, $(\XEL_t)_{t\geq 0}$ by \eqref{eq:def_coupling}.
% and $X_n = \XEL_{n\gastep}$ for $n\in\nset$.
Then there exists $C>0$ such that for all $n \in \nset$ and $\gastep\in\ocint{0,\gastep_0}$, almost surely,
\begin{multline*}
\CPE{\norm[2]{Y_{(n+1)\gamma}- \XEL_{(n+1)\gamma}}}{\filtration_{n\gamma}} \leq
\rme^{-\mU \gastep} \norm[2]{\XL_{n\gastep} - \XEL_{n\gastep}} \\
+ C\gastep^{2+\betah} \lyape(\XEL_{n\gastep})
+ C \gastep^3 \lyape(\XEL_{n\gastep}) \eqsp.
\end{multline*}
\end{lemma}

\begin{remark}
The calculations in the proof show that the dependence \wrt~$\XEL_{n\gastep}$ and $\XL_{n\gastep}$ is in fact polynomial but their exact expressions are very involved. For the sake of simplicity, we bound these polynomials by $\lyape$. The same remark applies equally to \Cref{lem:wasserstein-contraction-gamma2}.
\end{remark}

\begin{proof}
%If $\pU$ is twice continuously differentiable and there exist $\ell,\LG \geq 0 $ such that for all $x \in \rset^d$,
%\begin{equation}
%\label{lemma:prac-grad-loc-lipschitz}
%\norm{\nabla^2 U(x)} \leq \LG \defEns{1+ \norm[\ell]{x}} \eqsp,
%\end{equation}
%then \Cref{assum:grad-loc-lipschitz} is satisfied.
%Note that \Cref{assum:hessian} implies \Cref{assum:grad-loc-lipschitz} by \Cref{lem:reg-hessian-U-growth}-\ref{item:reg-hessian-U-growth-1} and \eqref{lemma:prac-grad-loc-lipschitz} with $\LG=\CH$ and $\ell = \nu+\betah$.
Note first that by \Cref{lem:reg-hessian-U-growth}-\ref{item:reg-hessian-U-growth-1}, \Cref{assum:hessian} implies \Cref{assum:grad-loc-lipschitz} with $\LG=\CH$ and $\ell = \nu+\betah$.
By the Markov property, we only need to show the result for $n=0$.
The proof is a refinement of \Cref{lem:wasserstein-contraction-gamma2} and we use the same notations.
We have to improve the bound on $A_1$ defined in \eqref{eq:A1A2}.
We decompose $A_1 = A_{11}+A_{12}$ where
\begin{align*}
A_{11} &= \expe{\ps{\Theta_t}{\nablaU(\XEL_t) - \nablaU(x) - \nablaS \pU(x)(\XEL_t -x)}} \eqsp,\eqsp \\
A_{12} &= \expe{\ps{\Theta_t}{\nablaS \pU(x)(\XEL_t -x)}} \eqsp.
\end{align*}
Using $\absolute{\ps{a}{b}} \leq (\mU/6)\norm[2]{a} + \{3/(2m)\} \norm[2]{b}$ for all $a,b\in\rset^d$,
\begin{equation}
\label{eq:A11}
\absolute{A_{11}} \leq \frac{m}{6} \expe{\norm[2]{\Theta_t}} + \frac{3}{2m} \expe{\norm[2]{\nablaU(\XEL_t) - \nablaU(x) - \nablaS \pU(x)(\XEL_t -x)}} \eqsp.
\end{equation}
By \Cref{lem:reg-hessian-U-growth}-\ref{item:reg-hessian-U-growth-2},
\begin{multline*}
%\label{eq:grad-Hess-U}
\norm[2]{\nablaU(\XEL_t) - \nablaU(x) - \nablaS \pU(x)(\XEL_t -x)} \\
\leq \frac{4 \LH^2}{(1+\betah)^2} \parenthese{1+\norm[\nue]{x}+\norm[\nue]{\XEL_t}}^2\norm[2(1+\betah)]{\XEL_t-x} \eqsp.
\end{multline*}
Following the proof of \Cref{lemma:kullback-leibler}, using \eqref{eq:temp-A11} and \eqref{eq:temp-A12}, we have
\begin{equation}
\label{eq:bound-taylor-order2}
\expe{\norm[2]{\nablaU(\XEL_t) - \nablaU(x) - \nablaS \pU(x)(\XEL_t -x)}} \leq t^{1+\betah} \polyga[4](\norm{x}) \eqsp.
\end{equation}
where $\polyga[4]:\rset_+ \to \rset_+$ is defined for all $s\in\rset_+$ by,
%\begin{align}
%\nonumber
%\polyga[4](s) &= \frac{4\LH^2}{(1+\betah)^2} \int_{\rset^d}
%\left[
%1+ s^\nue + \defEns{s + \gastep\parenthese{2\LG(1+s^{\ell+1}) + \gastep\Cal(1+s^\alphag)} + \sqrt{2\gastep}\norm{z}}^{\nue}
%\right]^2 \\
%\label{eq:def-polyga4}
%&\times
%\left[
%\sqrt{2}\norm{z} + \sqrt{\gastep}\defEns{2\LG(1+s^{\ell+1}) + \gastep\Cal(1+s^\alphag)}
%\right]^{2(1+\betah)}
%(2\uppi)^{-d/2} \rme^{-\norm[2]{z}/2} \rmd z \eqsp,
%\end{align}
\begin{align}
\nonumber
&\polyga[4](s) = \frac{4\LH^2}{(1+\betah)^2} \int_{\rset^d}
\left[
\sqrt{2}\norm{z} + \sqrt{\gastep}\defEns{2\LG(1+s^{\ell+1}) + \gastep\Cal(1+s^\alphag)}
\right]^{2(1+\betah)} \\
 \label{eq:def-polyga4}
&\times
\left[
1+ s^\nue + \defEns{s + \gastep\parenthese{2\LG(1+s^{\ell+1}) + \gastep\Cal(1+s^\alphag)} + \sqrt{2\gastep}\norm{z}}^{\nue}
\right]^2
\frac{\rme^{-\norm[2]{z}/2}}{(2\uppi)^{d/2}}
  \rmd z \eqsp.
\end{align}
We decompose $A_{12}$ in $A_{12} = A_{121}+A_{122}$ where
\begin{equation*}
A_{121} = \expe{\ps{\Theta_t}{-t \nablaS \pU(x) \Ugad(x)}} \eqsp,\eqsp
A_{122} = \sqrt{2} \expe{\ps{\Theta_t}{\nablaS \pU(x) B_t}} \eqsp.
\end{equation*}
Define $\polyga[5]:\rset_+\to\rset_+$ for $s\in\rset_+$ by,
\begin{equation}
\label{eq:def-polyga5}
\polyga[5](s) = \CH^2 \parenthese{1+s^{\nue+\betah}}^2
\defEns{2\LG(1+s^{\ell+1}) + \gastep\Cal (1+s^{\alphag})}^2 \eqsp.
\end{equation}
By \Cref{lem:reg-hessian-U-growth}-\ref{item:reg-hessian-U-growth-1} and \eqref{eq:estimate-drift-step},
\begin{equation}
\label{eq:A121}
\absolute{A_{121}} \leq (m/6) \expe{\norm[2]{\Theta_t}} + \{3/(2m)\} t^2 \polyga[5](\norm{x}) \eqsp.
\end{equation}
By Cauchy-Schwarz inequality and \Cref{lem:reg-hessian-U-growth}-\ref{item:reg-hessian-U-growth-1},
\begin{align}
\nonumber
\absolute{A_{122}} &= \sqrt{2} \absolute{\expe{\ps{\int_0^t \defEns{\nablaU(\XL_s) - \nablaU(y)}\rmd s}{\nablaS \pU(x) B_t}}} \\
\label{eq:A122-1}
&\leq \sqrt{2dt}\CH(1+\norm[\nue+\betah]{x})
\expe{\norm[2]{\int_0^t \defEns{\nablaU(\XL_s) - \nablaU(y)}\rmd s}}^{1/2} \eqsp.
%\expe{\norm[2]{\nablaS \pU(x) B_t}}^{1/2} \eqsp.
\end{align}
%By \Cref{lem:reg-hessian-U-growth}-\ref{item:reg-hessian-U-growth-1}, $\expe{\norm[2]{\nablaS \pU(x) B_t}}^{1/2} \leq \sqrt{dt}\CH(1+\norm[\nue+\betah]{x})$.
By \Cref{assum:grad-loc-lipschitz}, Cauchy-Schwarz inequality and using $(1+\norm[\ell]{y} + \norm[\ell]{\XL_s})^2 \leq 3 (2+\norm[2\ell]{y} + \norm[2\ceil{\ell}]{\XL_s})$ for $s\in\coint{0,\gastep}$, we have
\begin{multline*}
\expe{\norm[2]{\int_0^t \defEns{\nablaU(\XL_s) - \nablaU(y)}\rmd s}} \leq
3t\LG^2 (2 + \norm[2\ell]{y}) \int_{0}^{t} \expe{\norm[2]{\XL_s - y}} \rmd s  \\
+ 3t\LG^2 \int_{0}^{t} \expe{\norm[2\ceil{\ell}]{\XL_s}\norm[2]{\XL_s - y}} \rmd s \eqsp.
\end{multline*}
By \Cref{item:moment_diffusion_3_p,item:moment_diffusion_4_p}, we get
\begin{equation*}
\expe{\norm[2]{\int_0^t \defEns{\nablaU(\XL_s) - \nablaU(y)}\rmd s}} \leq
\frac{3 t^3 \LG^2}{2} \defEns{\parenthese{2+\norm[2\ell]{y}} \polyga[3](\norm{y})+\polygaq[\ceil{\ell}](\norm{y})} \eqsp,
\end{equation*}
where $\polyga[3], \polygaq[\ceil{\ell}] \in\Cpoly(\rset_+,\rset_+)$ are defined in \eqref{eq:def-polyga3} and \eqref{eq:polygaq}.
Plugging this result in \eqref{eq:A122-1}, we obtain
\begin{equation}
\label{eq:A122}
\absolute{A_{122}} \leq t^2 \sqrt{3d} \CH \LG \parenthese{1+\norm[\nue+\betah]{x}} \defEns{\parenthese{2+\norm[2\ell]{y}} \polyga[3](\norm{y})+\polygaq[\ceil{\ell}](\norm{y})}^{1/2} \eqsp.
\end{equation}
Combining \eqref{eq:A11}, \eqref{eq:bound-taylor-order2}, \eqref{eq:A121} and \eqref{eq:A122}, we get
\begin{align*}
\absolute{A_1} &\leq (\mU/3) \expe{\norm[2]{\Theta_t}} + \{3/(2\mU)\} \defEns{t^{1+\betah} \polyga[4](\norm{x}) + t^2 \polyga[5](\norm{x})} \\
&\phantom{--}+ t^2 \sqrt{3d} \CH \LG \parenthese{1+\norm[\nue+\betah]{x}} \defEns{\parenthese{2+\norm[2\ell]{y}} \polyga[3](\norm{y})+\polygaq[\ceil{\ell}](\norm{y})}^{1/2} \eqsp,
\end{align*}
and by \eqref{eq:A2},
$\absolute{A_2} \leq (m/6) \expe{\norm[2]{\Theta_t}} + \{3/(2\mU)\} \gastep^2 \polyp[2](\norm{x})$,
where $\polyp[2]\in\Cpoly(\rset_+,\rset_+)$ is defined in \eqref{eq:def-polyp2}.
 Combining these inequalities in \eqref{eq:derivative-thetat}, we get
\begin{multline*}
\frac{\rmd}{\rmd t} \expe{\norm[2]{\Theta_t}} \leq -m \expe{\norm[2]{\Theta_t}}
+ 3m^{-1} \defEns{\gastep^2 \polyp[2](\norm{x}) + t^{1+\betah} \polyga[4](\norm{x}) + t^2 \polyga[5](\norm{x})} \\
+ 2 t^2 \sqrt{3d} \CH \LG \parenthese{1+\norm[\nue+\betah]{x}} \defEns{\parenthese{2+\norm[2\ell]{y}} \polyga[3](\norm{y})+\polygaq[\ceil{\ell}](\norm{y})}^{1/2} \eqsp.
\end{multline*}
Using Gr\"onwall's lemma and $1-\rme^{-s} \leq s$ for all $s\geq 0$, we obtain
\begin{align*}
&\expe{\norm[2]{Y_{\gamma}- \XEL_{\gastep}}} \leq
\rme^{-\mU \gastep} \norm[2]{y - x} \\
&\phantom{--------}+ 3m^{-1} \defEns{\gastep^3 \polyp[2](\norm{x}) + \frac{\gastep^{2+\betah}}{2+\betah} \polyga[4](\norm{x}) + \frac{\gastep^3}{3} \polyga[5](\norm{x})} \\
&+ 2 \gastep^3 \sqrt{d/3} \CH \LG \parenthese{1+\norm[\nue+\betah]{x}} \defEns{\parenthese{2+\norm[2\ell]{y}} \polyga[3](\norm{y})+\polygaq[\ceil{\ell}](\norm{y})}^{1/2} \eqsp.
\end{align*}
Finally, by \eqref{eq:def-polyp2}, \eqref{eq:def-polyga3}, \eqref{eq:def-polyga4}, \eqref{eq:def-polyga5} and \eqref{eq:polygaq}, there exists $C>0$ such that for all $x\in\rset^d$ and $\gastep\in\ocint{0,\gastep_0}$,
\begin{align*}
&3m^{-1} \defEns{\gastep^3 \polyp[2](\norm{x}) + \frac{\gastep^{2+\betah}}{2+\betah} \polyga[4](\norm{x}) + \frac{\gastep^3}{3} \polyga[5](\norm{x})} \leq C \gastep^{2+\betah} \lyape(x) \eqsp, \\
&2 \sqrt{d/3} \CH \LG \parenthese{1+\norm[\nue+\betah]{x}} \leq C^{1/2} \lyape(x)^{1/2} \eqsp, \\
&\parenthese{2+\norm[2\ell]{x}} \polyga[3](\norm{x})+\polygaq[\ceil{\ell}](\norm{x}) \leq C \lyape(x) \eqsp.
\end{align*}
\end{proof}

\begin{proof}[Proof of \Cref{thm:wasserstein-gamma2}]
Let $\gastep\in\ocint{0,\gastep_0}$.
Define $(\XL_{t})_{t \geq 0}$, $(\XEL_t)_{t\geq 0}$ by \eqref{eq:def_coupling} and $X_n = \XEL_{n\gastep}$ for $n\in\nset$.
By \Cref{lem:wasserstein-contraction-gamma2} and \Cref{propo:drift_tamed_euler}, we have for all $n \in \nset$,
\begin{align}
\nonumber
&\expe{\norm[2]{Y_{n\gamma}- X_{n}}} \leq
\rme^{- n\mU \gastep} \norm[2]{y - x}
+ C \gastep^2 \sum_{k=0}^{n-1} \rme^{-\mU\gastep(n-1-k)}\expe{\lyape(\XE_k)} \\
\label{eq:iter-gamma2-wasserstein}
&\leq
\rme^{- n\mU \gastep} \norm[2]{y - x}
+  \frac{C\gastep^2}{1-\rme^{-\mU\gastep}}\frac{\be}{\ae^2}\rme^{\ae^2 \gastep}
+ C\gastep^2 \lyape(x) \sum_{k=0}^{n-1} \rme^{-\mU\gastep(n-1-k)} \rme^{-\ae^2 \gastep k} \eqsp.
\end{align}
%Define $\mn = \min(\rme^{-\mU}, \rme^{-\ae^2})$ and $\mx = \max(\rme^{-\mU}, \rme^{-\ae^2})$.
Note that
\begin{equation*}
\sum_{k=0}^{n-1} \rme^{-\mU\gastep(n-1-k)} \rme^{-\ae^2 \gastep k}
\leq \frac{n}{1-\max(\rme^{-\mU}, \rme^{-\ae^2})^{\gastep}} \eqsp
\end{equation*}
and $1-s^\gastep \geq -\gastep \log(s) \rme^{\gastep\log(s)}$ for $s\in\ooint{0,1}$.
%\begin{equation*}
%1-\rme^{-\mU\gastep} \geq \mU\gastep\rme^{-\mU\gastep} \quad , \quad
%1-\max(\rme^{-\mU}, \rme^{-\ae^2})^{\gastep} \geq \gastep\log(1/\max(\rme^{-\mU}, \rme^{-\ae^2}))\rme^{\gastep\log(\max(\rme^{-\mU}, \rme^{-\ae^2}))} \eqsp.
%\end{equation*}
In eq.~\eqref{eq:iter-gamma2-wasserstein}, integrating $y$ with respect to $\invpi$, for all $n\in\nset$, $(\XL_{n\gastep}, \XE_n)$ is a coupling between $\invpi$ and $\delta_x \rker^n$. By \Cref{theo:wasser_p_diffusion}, we get \eqref{eq:thm-wasserstein2-1}.
By \Cref{propo:drift_tamed_euler} and \cite[Corollary 6.11]{VillaniTransport}, we have for all $x\in\rset^d$, $\lim_{n\to\plusinfty} W_2(\delta_x \rker^n, \invpi) = W_2(\invpig, \invpi)$ and we obtain \eqref{eq:thm-wasserstein2-2}.
\end{proof}

\begin{proof}[Proof of \Cref{thm:wasserstein-gamma3}]
%Note that \Cref{assum:hessian} implies \Cref{assum:grad-loc-lipschitz} by \Cref{lem:reg-hessian-U-growth}-\ref{item:reg-hessian-U-growth-1} and \Cref{lemma:prac-grad-loc-lipschitz} with $\LG=\CH$ and $\ell = \nu+\betah$.
Let $\gastep\in\ocint{0,\gastep_0}$.
Define $(\XL_{t})_{t \geq 0}$, $(\XEL_t)_{t\geq 0}$ by \eqref{eq:def_coupling} and $X_n = \XEL_{n\gastep}$ for $n\in\nset$.
By \Cref{lem:wasserstein-contraction-gamma3}, we have for all $n \in \nset$,
\begin{equation*}
\expe{\norm[2]{Y_{n\gamma}- X_{n}}} \leq
\rme^{- n\mU \gastep} \norm[2]{y - x} + A_n + B_n \eqsp,
\end{equation*}
where
\begin{align*}
A_n &= C \gastep^{2+\betah} \sum_{k=0}^{n-1} \rme^{-\mU \gastep (n-1-k)} \expe{\lyape(\XE_k)} \eqsp, \\
B_n &= C \gastep^{3} \sum_{k=0}^{n-1} \rme^{-\mU \gastep (n-1-k)} \expe{\lyape(\XL_{k\gastep})} \eqsp.
\end{align*}
Analysis similar to the proof of \Cref{thm:wasserstein-gamma2} using \Cref{prop:existence_ergodicity} instead of \Cref{propo:drift_tamed_euler} for $B_n$ shows then the result.
%Integrating $y$ with respect to $\invpi$, for all $n\in\nset$, $(\XL_{n\gastep}, \XE_n)$ is a coupling between $\invpi$ and $\delta_x \rker^n$ and by \Cref{theo:wasser_p_diffusion}, we get the result.
\end{proof}

\subsection{Proof of \Cref{prop:bias-MSE-poisson-equation}}
\label{sec:proof-prop:bias-MSE-poisson-equation}

We first state a lemma on the existence and regularity of a solution of the Poisson equation \eqref{eq:eq-Poisson} which is adapted from \cite[Theorem 1]{pardoux2001}.

\begin{lemma}
\label{prop:existence-sol-Poisson}
Assume \Cref{assum:ergodicity-U} and \Cref{assumption:U-C4}. Let $\tf\in\Csetfunction^3(\rset^d,\rset)$ be such that $\norm{\DD^i \tf} \in \Cpoly(\rset^d, \rset_+)$ for $i\in\defEns{0,\ldots,3}$.
Then, there exists a solution of the Poisson equation \eqref{eq:eq-Poisson} $\sPoi\in\Csetfunction^4(\rset^d,\rset)$, such that $\norm{\DD^i \sPoi} \in\Cpoly(\rset^d,\rset_+)$  for $i\in\defEns{0,\ldots,4}$.
\end{lemma}

\begin{proof}
The proof is postponed to \Cref{sec:proof-prop-sol-Poisson}.
\end{proof}

\begin{proof}[Proof of \Cref{prop:bias-MSE-poisson-equation}]
The proof is adapted from \cite[Section 5.1]{doi:10.1137/090770527}
%The proof is adapted from \cite[Section 5.1]{doi:10.1137/090770527}.
Let $\gastep\in\ocint{0,\gastep_0}$.
In this Section, $C$ is a positive constant which can change from line to line but does not depend on $\gastep$.
For $k\in\nset$, denote by
%$\sPoi_k, \tf_k, \nablaU_k, \dr[k]$ the quantities $\sPoi(X_k), \tf(X_k), \nabla \pU (X_k), \dr(X_k)$ and
\[ \inc_{k+1} = X_{k+1} - X_k = -\gastep \dr(X_k) + \sqrt{2\gastep} \nZ_{k+1} \eqsp. \]
%Note that \Cref{assumption:U-C4} implies \Cref{assum:grad-loc-lipschitz}.
By \Cref{assum:ergodicity-U}, \Cref{assumption:U-C4} and \Cref{prop:existence-sol-Poisson}, there exists a solution to the Poisson equation \eqref{eq:eq-Poisson} $\sPoi\in\Csetfunction^4(\rset^d,\rset)$, such that for all $x\in\rset^d$ and $i\in\defEns{0,\ldots,4}$,
\begin{equation}\label{eq:prop-poisson-equation}
  \generator \sPoi(x) = -\parenthese{\tf(x) -  \invpi(\tf)} \quad \text{and} \quad
\norm{\DD^i \sPoi} \in\Cpoly(\rset^d,\rset_+) \eqsp.
\end{equation}
%By \Cref{prop:existence-sol-Poisson}, there exists
By Taylor's formula, we have for $k\in\nset$,
\begin{align*}
\sPoi(X_{k+1}) &= \sPoi(X_k) + \DD \sPoi(X_k) [\inc_{k+1}] + (1/2)\DD^2 \sPoi(X_k) [\inc_{k+1}, \inc_{k+1}] \\
&\phantom{----------}+ (1/6)\DD^3\sPoi(X_k)[\inc_{k+1},\inc_{k+1},\inc_{k+1}] + \rr_k \eqsp,\\
\rr_k &= (1/6) \int_0^1 (1-s)^3 \DD^4 \sPoi(X_k + s\inc_{k+1}) [\inc_{k+1}, \inc_{k+1},\inc_{k+1},\inc_{k+1}] \rmd s \eqsp.
\end{align*}
Using the expression of $\inc_{k+1}$ and \eqref{eq:generator-langevin}, we get
\begin{align*}
&\sPoi(X_{k+1}) = \sPoi(X_k) + \gaStep \generator \sPoi(X_k) + \sqrt{2\gaStep} \DD \sPoi(X_k) [\nZ_{k+1}] \\
&\phantom{--}+ \gaStep \defEns{\DD^2 \sPoi(X_k) [\nZ_{k+1},\nZ_{k+1}] - \Delta \sPoi(X_k)} + \gaStep \DD \sPoi(X_k)[\nablaU(X_k) - \dr(X_k)] \\
&\phantom{--}+ (\gaStep^2/2) \DD^2 \sPoi(X_k) [\dr(X_k),\dr(X_k)] -\sqrt{2}\gaStep^{3/2} \DD^2 \sPoi(X_k) [\dr(X_k), \nZ_{k+1}] \\
&\phantom{--}+ (1/6)\DD^3 \sPoi(X_k) [\inc_{k+1},\inc_{k+1},\inc_{k+1}] + r_k \eqsp.
\end{align*}
Summing from $k=0$ to $n-1$ for $n\in\nset^\star$, dividing by $n\gaStep$,
%and by \Cref{assumption:sol-poisson}($\tf, \sPoi$)
we get
\begin{equation*}
\frac{1}{n}\sum_{k=0}^{n-1} (\tf(X_k) - \invpi(\tf)) = \frac{\sPoi(X_0) - \sPoi(X_n)}{n\gaStep} + \frac{1}{n\gaStep}\parenthese{\sum_{i=0}^3 \MM_{i,n} + \sum_{i=0}^3 \SS_{i,n}} \eqsp,
\end{equation*}
where
\begin{align*}
&\MM_{0,n} = ((\sqrt{2}\gaStep^{3/2})/6) \sum_{k=0}^{n-1} \big\{2 \DD^3 \sPoi(X_k) [\nZ_{k+1},\nZ_{k+1},\nZ_{k+1}] \\
&\phantom{----------------}+ 3\gaStep \DD^3 \sPoi(X_k)[\dr(X_k),\dr(X_k), \nZ_{k+1}]
\big\} \eqsp, \\
&\MM_{1,n} = \gaStep \sum_{k=0}^{n-1} (\DD^2 \sPoi(X_k) [\nZ_{k+1},\nZ_{k+1}] - \Delta \sPoi(X_k) ) \eqsp, \\
&\MM_{2,n} = \sqrt{2\gaStep} \sum_{k=0}^{n-1} \DD \sPoi(X_k) [\nZ_{k+1}] \eqsp,\\
&\MM_{3,n} = -\sqrt{2} \gaStep^{3/2} \sum_{k=0}^{n-1} \DD^2 \sPoi(X_k) [\dr(X_k), \nZ_{k+1}] \eqsp,
\end{align*}
and
\begin{align*}
&\SS_{0,n} = - (\gaStep^2/6) \sum_{k=0}^{n-1} \big\{6 \DD^3 \sPoi(X_k) [\dr(X_k), \nZ_{k+1},\nZ_{k+1}] \\
&\phantom{----------------}+ \gastep \DD^3 \sPoi(X_k) [\dr(X_k),\dr(X_k),\dr(X_k)]
\big\} \eqsp, \\
&\SS_{1,n} = \gaStep \sum_{k=0}^{n-1} \DD \sPoi(X_k)[\nablaU(X_k) - \dr(X_k)] \eqsp, \\
&\SS_{2,n} = (\gaStep^2/2) \sum_{k=0}^{n-1} \DD^2 \sPoi(X_k) [\dr(X_k),\dr(X_k)] \eqsp, \\
&\SS_{3,n} = \sum_{k=0}^{n-1} \rr_k \eqsp.
\end{align*}
%$\MM_{0,n}$ and $\SS_{0,n}$ follow from the decomposition of $(1/6) \sum_{k=0}^{n-1} \DD^3 \sPoi(X_k) [\inc_{k+1},\inc_{k+1},\inc_{k+1}]$.
By \Cref{assum:dr-close-nablaU}, we calculate for $n\in\nset^*$, $\absolute{\SS_{1,n}} \leq \gastep^2 \Cal \sum_{k=0}^{n-1} \norm{\DD \sPoi(X_k)}(1+\norm[\alphag]{\XE_k})$. By \Cref{assumption:U-C4}, \eqref{eq:estimate-drift-step} and \eqref{eq:prop-poisson-equation}, there exist $p,q \geq 1$ and $C_q >0$ such that the summands of $(\MM_{i,n})_{n\in\nset}$ and $(\SS_{i,n})_{n\in\nset}$ for $i\in\defEns{0,\ldots,3}$ are dominated by $C_q\parenthese{1+\norm[q]{X_k}}(1+\norm[p]{\nZ_{k+1}})$ for $k\in\defEns{0,\ldots,n-1}$. Therefore, by \Cref{prop:existence_ergodicity}, for $i\in\defEns{0,\ldots,3}$, $(\MM_{i,n})_{n\in\nset}$ are martingales and for $n\in\nset^*$, $\expe{\SS_{i,n}^2} \leq C n^2 \gastep^4$,
%\begin{equation*}
%\expe{\absolute{\SS_{0,n}}}+\expe{\absolute{\SS_{1,n}}}+\expe{\absolute{\SS_{2,n}}}+\expe{\absolute{\SS_{3,n}}}
%\leq C n \gastep^2
%\end{equation*}
%which yields \eqref{eq:bias-poisson-eq}.
%Similarly, for $n\in\nset^*$,
\begin{equation*}
\expe{\MM_{0,n}^2} \leq C n \gastep^3 \eqsp,\eqsp
\expe{\MM_{1,n}^2} \leq C n \gastep^2 \eqsp,\eqsp
\expe{\MM_{2,n}^2} \leq C n \gastep \eqsp,\eqsp
\expe{\MM_{3,n}^2} \leq C n \gastep^3 \eqsp,
\end{equation*}
which yield the result.
%and,
%\begin{equation*}
%\expe{\SS_{0,n}^2}+\expe{\SS_{1,n}^2}+\expe{\SS_{2,n}^2}+\expe{\SS_{3,n}^2}
%\leq C n^2 \gastep^4 \eqsp,
%\end{equation*}
%which results in \eqref{eq:MSE-poisson-eq}.
\end{proof}

%% file: appendix.tex
\section{Proof of \Cref{lem:moment_diffusion_2_p}}
\label{appendix:proof-lem-moment}

  By \Cref{assum:convex}, \eqref{eq:langevin} has a unique strong
  solution $(\XL_t)_{t\geq 0}$ for any initial data $Y_0=x \in \rset^d$.
  Define for $p \in\nset^*$, $V_p:\rset^d \to \rset_+$ by $\VStar_p(y) =
  \norm[2p]{y}$ for $y\in\rset^d$. We have using \Cref{assum:convex},
\begin{align}
  \label{eq:drift-Vp-0}
  \generator V_p (x) &= -2p \norm[2(p-1)]{x}\ps{\nabla U(x)}{x} + 2p(d + 2(p-1))  \norm[2(p-1)]{x}\\
  \label{eq:drift-Vp}
&\leq -2p m \norm[2p]{x} + 2p\norm[2(p-1)]{x}(d + 2(p-1)) \eqsp.
\end{align}
Applying \cite[Theorem 1.1]{meyn:tweedie:1993:III} with $V(x,t)=V_p(x)\rme^{2pm t}$, $g_{-}(t) = 0$ and $g_{+}(x,t) = 2p(d+2(p-1)) V_{p-1}(x) \rme^{2pm t}$ for $x\in\rset^d$ and $t\geq 0$, we get denoting by $v_p(t,x) = P_t V_p(x)$,
%By  \cite[Chapter VII, proposition 1.2]{revuz:yor:1999}, , we have
%\begin{equation*}
%  \frac{\partial v_p(t,x)}{\partial t} \leq - 2pm v_p(t,x) + 2p(d + 2(p-1)) v_{p-1}(t,x) \eqsp.
%\end{equation*}
%By Grönwall's inequality, we get for all $x \in \rset^d$ and $t \geq 0$,
\begin{equation*}
  v_p(t,x)
 \leq \rme^{-2pmt} V_p(x) +  2p(d + 2(p-1)) \int_0^t \rme^{-2pm(t-s)} v_{p-1}(s,x) \rmd s \eqsp.
\end{equation*}
A straightforward induction concludes the proof.

\section{Proof of \Cref{theo:wasser_p_diffusion}}
\label{appendix:proof-theo-wasser}

By \Cref{eq:drift-Vp} and \cite[Theorem 2.2]{roberts:tweedie:1996}, $(\XL_t)_{t \geq 0}$ the solution of
\eqref{eq:langevin} is $V_p$-geometrically ergodic \wrt~$\invpi$. Taking the limit $t\to\plusinfty$ in \Cref{lem:moment_diffusion_2_p} concludes the proof.

\section{Proof of \Cref{item:moment_diffusion_3_p}}
\label{appendix:proof-item-moment}

%The proof is made by induction on $p$.
%First assume that $p=1$.
Define $\tildeVx:\rset^d\to\rset_+$ for all $y\in\rset^d$ by $\tilde{V}_x(y) =
\norm[2]{y-x}$.
By \Cref{lem:moment_diffusion_2_p}, the process
$(\tilde{V}_x(\XL_t^{}) -\tilde{V}_x(x) - \int_0 ^t \generator
\tilde{V}_x(Y^{}_s) \rmd s)_{ t \geq 0}$, is a $(\mathcal{F}_t)_{ t
  \geq 0}$-martingale. Denote for all $t \geq 0$ and $y \in \rset^d$ by
$\tilde{v}(t,x) = P_t \tilde{V}_x(x)$. Then we get,
\begin{equation}
\label{eq:Dynkin0_2}
\frac{\partial \tilde{v}(t,x)}{\partial t} = P_t \generator \tilde{V}_x(x) \eqsp.
\end{equation}
By \Cref{assum:convex}, we have for all $y \in \rset^d$,
\begin{equation}
\label{eq:generator_2}
\generator \tilde{V}_x(y) = 2 \parenthese{ - \ps{\nabla U(y)}{y-x} +d }
\leq 2  \parenthese{ -m \tilde{V}_x(y) +d - \ps{\nabla U(x) } {y-x}} \eqsp.
\end{equation}
Using \eqref{eq:Dynkin0_2}, this inequality and that $\tilde{V}_x$ is nonnegative, we get
\begin{equation}
\label{eq:preparation_Grown}
\frac{\partial \tilde{v}(t,x)}{\partial t} = P_t \generator \tilde{V}_x(x)   \leq 2\parenthese{ d -
  \int_{ \rset^d}  \ps{\nabla U(x) } {y-x} P_t( x , \rmd y)}  \eqsp.
\end{equation}
Using \eqref{eq:estimate_nabla_u_x} and \eqref{eq:langevin}, %and Jensen's inequality,
we have
\begin{align}
\nonumber
\abs{\expeMarkov{x}{ \ps{\nabla U(x) } {\XL_t^{}-x}}} &\leq \norm{\nabla U(x) } \norm{\expeMarkov{x}{\XL_t^{}-x}} \\
\nonumber
&\leq \norm{\nabla U(x)} \norm{\expeMarkov{x}{\int_0^t \left\{ \nabla U(\XL_s^{})  \right\} \rmd s}} \\
\label{lem:moment_short_time_proof_1}
&  \leq  2\LG\defEns{1+\norm[\ell+1]{x}} \int_0^t\PE_x\left[ \norm{\nabla U(\XL_s^{})  }\right]\rmd s   \eqsp.
\end{align}
%Denote $\N = \ceil{(\ell+1)/2}$.
Using \eqref{eq:estimate_nabla_u_x} again,
\begin{align}
\nonumber
  \int_0^t\PE_x\left[ \norm{\nabla U(\XL_s^{}) }\right]\rmd s &\leq   2\LG \int_0^t\PE\left[ 1+\norm[\ell+ 1]{\XL_s}   \right]\rmd s \\
  \label{lem:moment_short_time_proof_2}
&\leq 2\LG \defEns{2t  + \int_0^t\PE\left[ \norm[2\N]{\XL_s}
    \right] \rmd s } \eqsp.
\end{align}
Furthermore using that for all $s \geq 0 $, $1-\rme^{-s} \leq s$, $s+\rme^{-s}-1 \leq  s^2/2$, and \Cref{lem:moment_diffusion_2_p} we get
\begin{align*}
  \int_0^t\PE_x\left[ \norm[2\N ]{\XL_s} \right] \rmd s & \leq a_{0,\N}  \frac{2\N tm + \rme^{-2\N mt}-1}{2\N m}  +\sum_{k=1}^\N a_{k,\N} \norm[2k]{x} \frac{1-\rme^{-2mkt}}{2km}
\\
&\leq t^2 \N m a_{0,\N} +  t \sum_{k=1}^\N a_{k,\N} \norm[2k]{x}   \eqsp.
\end{align*}
Plugging this inequality in \eqref{lem:moment_short_time_proof_2} and \eqref{lem:moment_short_time_proof_1}, we get
\begin{equation}
\label{lem:moment_short_time_proof_4}
  \abs{\expeMarkov{x}{ \ps{\nabla U(x) } {\XL_t^{}-x}}}
\leq 4\LG^2(1+\norm[\ell+1]{x}) \defEns{2t+\N \mU a_{0,\N} t^2 + t \sum_{k=1}^\N a_{k,\N} \norm[2k]{x}}  \eqsp.
\end{equation}
Using this bound in \eqref{eq:preparation_Grown} and integrating the inequality gives
\begin{equation}
\label{eq:details-Yt-x-diffusion}
\tilde{v}(t,x) \leq   2d t + 8\LG^2(1+\norm[\ell+1]{x}) \defEns{t^2+\N \mU a_{0,\N}(t^3/3) + (t^2/2)\sum_{k=1}^\N a_{k,\N}\norm[2k]{x}} \eqsp.
  \end{equation}

\section{Proof of \Cref{item:moment_diffusion_4_p}}
\label{appendix:proof-item-diffusion}

We show the result by induction on $p$.
% \begin{align}
% \nonumber
%&\expeMarkov{x}{\norm[2p]{\XL_t}\norm[2]{\XL_t - x}} \leq
%\defEns{\prod_{i=1}^p 2i(d+3i-2)}
%\Bigg[
%2d\frac{t^{p+1}}{(p+1)!} + 8\LG^2(1+\norm[\ell+1]{x}) \\
%\nonumber
%&\phantom{-------}\times
%	\bigg\{
%	(2+\sum_{k=1}^N a_{k,N} \norm[2k]{x}) \frac{t^{p+2}}{(p+2)!} + 2N m a_{0,N} \frac{t^{p+3}}{(p+3)!}
%	\bigg\}
%\Bigg] \\
%\nonumber
%&\phantom{---} + 2\sum_{k=1}^p \defEns{\prod_{i=k+1}^p 2i(d+3i-2)} \defEns{d+4+\frac{\LG^2(1+\norm[\ell+1]{x})^2}{m(k+1)}} \\
%\label{eq:diffusion-Yt-Yt-x}
%&\phantom{-------} \times
%\defEns{\parenthese{\sum_{i=1}^k a_{i,k} \norm[2i]{x}} \frac{t^{p+1-k}}{(p+1-k)!} + 2km a_{0,k} \frac{t^{p+2-k}}{(p+2-k)!}} \eqsp.
% \end{align}
The case $p=0$ follows from \eqref{eq:details-Yt-x-diffusion}. Suppose $p\geq 1$.
Define for $y\in\rset^d$, $W_{x,p}:\rset^d\to\rset_+$ by $W_{x,p}(y) = \norm[2p]{y}\norm[2]{y-x}$. We have
\begin{multline*}
\generator W_{x,p}(y) = -2\norm[2p]{y} \ps{\nablaU(y)}{y-x} - (2p) \norm[2(p-1)]{y} \norm[2]{y-x} \ps{\nablaU(y)}{y} \\
+2 \norm[2(p-1)]{y} \defEns{d\norm[2]{y} + 4p\ps{y}{y-x} + p(d+2p-2)\norm[2]{y-x}} \eqsp.
\end{multline*}
By \Cref{assum:convex}, \eqref{eq:estimate_nabla_u_x} and using $\absolute{\ps{a}{b}} \leq \eta \norm[2]{a} + (4\eta)^{-1}\norm[2]{b}$ for all $\eta>0$, we have
\begin{align}
\nonumber
&\generator W_{x,p}(y) \leq  \frac{\norm[2p]{y}\norm[2]{\nablaU(x)}}{2m(p+1)}
+ 2\norm[2(p-1)]{y} \defEns{(d+4)\norm[2]{y} + p(d+3p-2)\norm[2]{y-x}} \\
\label{eq:generator-Wxp}
&\leq \norm[2p]{y} \defEns{2(d+4) + \frac{2\LG^2(1+\norm[\ell+1]{x})^2}{m(p+1)}}
+2p(d+3p-2)\norm[2]{y-x} \norm[2(p-1)]{y} \eqsp.
\end{align}
By \Cref{lem:moment_diffusion_2_p}, the process
$(W_{x,p}(\XL_t^{}) -W_{x,p}(x) - \int_0 ^t \generator
W_{x,p}(Y^{}_s) \rmd s)_{ t \geq 0}$ is a $(\mathcal{F}_t)_{ t
  \geq 0}$-martingale.
For $x\in\rset^d$ and $t\geq 0$, denote by $w_{x,p}(x,t)= P_t W_{x,p}(x)$ and $v_p(x,t) = \expeMarkov{x}{\norm[2p]{\XL_t}}$.
Taking the expectation of \eqref{eq:generator-Wxp} \wrt~$\delta_x \sgP_t$ and integrating \wrt~$t$, we get
%\begin{equation}
%\label{eq:Dynkin0_3}
%\frac{\partial \tilde{v}(t,x)}{\partial t} = P_t \generator \tilde{V}_x(x) \eqsp.
%\end{equation}
\begin{align*}
&w_{x,p}(t,x) \leq 2 \defEns{d+4 + \frac{\LG^2(1+\norm[\ell+1]{x})^2}{m(p+1)}} \int_0^t v_p(s,x) \rmd s \\
&\phantom{------------}+ 2p(d+3p-2) \int_0^t w_{x,p-1}(s,x) \rmd s \eqsp.
\end{align*}
By \Cref{lem:moment_diffusion_2_p}, $v_p(t,x) \leq 2pm a_{0,p} t + \sum_{k=1}^p a_{k,p} \norm[2k]{x}$. A straightforward induction concludes the proof. %of \eqref{eq:diffusion-Yt-Yt-x}.
%For $p\in\nset$, defining $\polygaq[p]\in\Cpoly(\rset_+,\rset_+)$ for $s\in\rset_+$ by,
%\begin{align}
%\nonumber
%&\polygaq[p](s) =
%\defEns{\prod_{i=1}^p 2i(d+3i-2)}
%\Bigg[
%2d\frac{\gastep^{p}}{(p+1)!} + 8\LG^2(1+s^{\ell+1}) \\
%\nonumber
%&\phantom{-------}\times
%	\bigg\{
%	(2+\sum_{k=1}^N a_{k,N} s^{2k}) \frac{\gastep^{p+1}}{(p+2)!} + 2N m a_{0,N} \frac{\gastep^{p+2}}{(p+3)!}
%	\bigg\}
%\Bigg] \\
%\nonumber
%&\phantom{---} + 2\sum_{k=1}^p \defEns{\prod_{i=k+1}^p 2i(d+3i-2)} \defEns{d+4+\frac{\LG^2(1+s^{\ell+1})^2}{m(k+1)}} \\
%\label{eq:polygaq}
%&\phantom{-------} \times
%\defEns{\parenthese{\sum_{i=1}^k a_{i,k} s^{2i}} \frac{\gastep^{p-k}}{(p+1-k)!} + 2km a_{0,k} \frac{\gastep^{p+1-k}}{(p+2-k)!}} \eqsp,
%\end{align}
%we get the result.

%\section{Poisson equation}
%
%\begin{proposition}
%\label{prop:existence-sol-Poisson}
%Assume \Cref{assum:ergodicity-U} and \Cref{assumption:U-C4}. Let $\tf\in\Csetfunction^3(\rset^d,\rset)$ be such that $\norm{\DD^i \tf} \in \Cpoly(\rset^d, \rset_+)$ for $i\in\defEns{0,\ldots,3}$.
%Then, there exists a solution of the Poisson equation \eqref{eq:eq-Poisson} $\sPoi\in\Csetfunction^4(\rset^d,\rset)$, such that $\norm{\DD^i \sPoi} \in\Cpoly(\rset^d,\rset_+)$  for $i\in\defEns{0,\ldots,4}$.
%\end{proposition}
%
%\begin{proof}
%The proof is postponed to \Cref{sec:proof-prop-sol-Poisson}.
%\end{proof}

\section{Proof of \Cref{prop:existence-sol-Poisson}}
\label{sec:proof-prop-sol-Poisson}

The proof is adapted from \cite[Theorem 1]{pardoux2001} and follows the same steps. Define $\tfb = \tf - \invpi(\tf)$. Note that \Cref{assumption:U-C4} implies \Cref{assum:grad-loc-lipschitz}. By \Cref{assum:ergodicity-U},
%\cite[Theorem 10.2.1, Corollary 10.1.5]{stroock2007multidimensional},
\cite[Corollary 11.1.5]{stroock2007multidimensional},
%$\defEns{\sgP(x,\cdot), x\in\rset^d}$ forms a strong Markov family and $\defEns{\sgP_t(x,\cdot), (t,x)\in\rset_+ \times \rset^d}$
$(\sgP_t)_{t\geq 0}$ is Feller continuous, which implies that for all $t>0$, if $(x_n)_{n\in\nset}$ is a sequence in $\rset^d$ converging to $x\in\rset^d$, then $\delta_{x_n} \sgP_t$ weakly converges to $\delta_x \sgP_t$.
Therefore, for all $t>0$ and $K>0$, $x\mapsto \sgP_t (\tf \vee (-K) \wedge K)(x)$ is continuous.
By Cauchy-Schwarz and Markov's inequalities, for all $t, K>0$ and $x\in\rset^d$, we have
\begin{align*}
  \absolute{\sgP_t (\tf \vee (-K) \wedge K)(x) - \sgP_t f(x)} & \leq \sgP_t (\absolute{f} \1\defEns{\absolute{f}\geq K})(x) \\
   & \leq \sgP_t f^2 (x) / K
\end{align*}
By \Cref{prop:existence_ergodicity} and the polynomial growth of $f$, we get for all $R>0$, 
\[
\lim_{K\to\plusinfty} \sup_{\norm{x}\leq R} \absolute{\sgP_t (\tf \vee (-K) \wedge K)(x)-\sgP_t(\tf)(x)} =0 
\] 
and therefore $x \mapsto \sgP_t \tfb (x)$ is continuous for all $t>0$.

By \eqref{eq:drift-Vp-0} and \cite[Theorem 3.10, Section 4.1]{DOUC2009897},
%By \cite[Proof of Theorem 1, step (a)]{pardoux2001}
 there exist $C,\varsigma>0$ and $p\in\nset$ such that for all $x\in\rset^d$ and $N>0$,
\[ \int_N^{\plusinfty} \absolute{\sgP_t \tfb(x)} \rmd t \leq C \parenthese{1+\norm[p]{x}} N^{-\varsigma} \eqsp. \]
Therefore, we may define $\sPoi(x) = \int_0^{\plusinfty} \sgP_t \tfb(x) \rmd t $ for all $x\in\rset^d$.
Denote by $\sPoiN = \int_0^N \sgP_t \tfb(x) \rmd t$ for all $N>0$ and $x\in\rset^d$. We have $\lim_{N\to\plusinfty} \sPoiN(x) = \sPoi(x)$ locally uniformly in $x$ and by continuity of $\sPoiN$ for all $N>0$, $\sPoi\in\Cpoly(\rset^d,\rset)$.

Let $x\in\rset^d$ and consider the Dirichlet problem,
\begin{equation*}
\generator \sPoihat(y) = -\tfb(y) \quad \text{for} \quad y\in\boule{x}{1} \quad \text{and} \quad
\sPoihat(y) = \sPoi(y) \quad\text{for}\quad y\in\partial\boule{x}{1} \eqsp,
\end{equation*}
where $\partial \boule{x}{1} = \boulefermee{x}{1} \setminus \boule{x}{1}$.
By \cite[Lemma 6.10, Theorem 6.17]{gilbarg2015elliptic}, there exists a solution $\sPoihat\in\Csetfunction^4(\boule{x}{1},\rset) \cap \Csetfunction(\boulefermee{x}{1}, \rset)$.
Let $\tilde{x}\in\boulefermee{x}{1/2}$.
By \Cref{assum:ergodicity-U}, \eqref{eq:langevin} has a unique strong solution denoted $(\XL^{\tilde{x}}_t)_{t\geq 0}$ starting at $\XL_0 = \tilde{x}$. Define the stopping time $\tau = \inf\defEns{t\geq 0 : \XL^{\tilde{x}}_t \notin \boule{x}{1}}$.
%\cite[Volume I, Chapter 6, equation (5.11)]{friedman2012stochastic}
By \cite[Volume I, Chapter 6, Theorem 5.1]{friedman2012stochastic}, we have
\begin{equation*}
\sPoihat(\tilde{x}) = \expe{\sPoi(\XL^{\tilde{x}}_{\tau})} + \expe{\int_0^{\tau} \tfb(\XL^{\tilde{x}}_t) \rmd t} \eqsp.
\end{equation*}
%and by \cite[Proof of Theorem 1, step (d)]{pardoux2001}, we get $\sPoi(x) = \sPoihat(x)$.
For all $N>0$, we decompose $\sPoiN(\tilde{x}) = A_N + B_N$ where
\begin{equation*}
  A_N = \int_0^N \expe{\tfb(\XL^{\tilde{x}}_t) \1\defEns{t\leq \tau}} \rmd t \quad , \quad
  B_N = \int_0^N \expe{\tfb(\XL^{\tilde{x}}_t) \1\defEns{t>\tau}} \rmd t \eqsp.
\end{equation*}
Since $\expe{\tau}<\plusinfty$ by \cite[Volume I, Chapter 6, equation (5.11)]{friedman2012stochastic},
\begin{equation*}
  \expe{\int_0^{\plusinfty} \absolute{\tfb(\XL^{\tilde{x}}_t)} \1\defEns{t \leq \tau} \rmd t} < \plusinfty \eqsp,
\end{equation*}
and by Fubini's theorem and the dominated convergence theorem, $\lim_{N\to\plusinfty} A_N = \expe{\int_0^\tau \tfb(\XL^{\tilde{x}}_t) \rmd t}$. We also have
\begin{equation*}
  B_N = \expe{\int_0^{(N-\tau)_+} \tfb(\XL^{\tilde{x}}_{\tau+t}) \rmd t}
   = \expe{\sPoi_{(N-\tau)_+} (\XL^{\tilde{x}}_\tau)} \eqsp.
\end{equation*}
Since $\expe{\tau}<\plusinfty$, we have $\lim_{N\to\plusinfty} \sPoi_{(N-\tau)_+} (\XL^{\tilde{x}}_\tau) = \sPoi(\XL^{\tilde{x}}_\tau)$ almost surely. Besides, there exist $C,p>0$ such that $\sPoi_{\tilde{N}}(\XL^{\tilde{x}}_\tau) \leq C(1+\norm[p]{x})$ almost surely and for all $\tilde{N}\geq 0$ because $\XL^{\tilde{x}}_\tau \in \boulefermee{x}{1}$ and $\sPoi_{\tilde{N}}$ converges locally uniformly to $\sPoi$.
By the dominated convergence theorem, we get $\lim_{N\to\plusinfty} B_N = \expe{\sPoi(\XL^{\tilde{x}}_\tau)}$. Taking the limit $\N\to\plusinfty$ of $\sPoiN(\tilde{x}) = A_N + B_N$, we obtain $\sPoi(\tilde{x}) = \sPoihat(\tilde{x})$.

Finally, by \cite[Problem 6.1 (a)]{gilbarg2015elliptic}, we obtain $\norm{\DD^i \sPoi}\in\Cpoly(\rset^d,\rset_+)$ for $i\in\defEns{0,\ldots,4}$ which concludes the proof.

\section{Badly conditioned multivariate Gaussian variable}
\label{sec:suppl-badly-conditioned-gaussian}

In this example, we consider a badly conditioned multivariate Gaussian variable in dimension $d=100$, of mean $0$ and covariance matrix $ \diag(10^{-5}, 1 \ldots, 1)$. We run $100$ independent simulations of ULA and TULAc, starting at $0$, with a step size $\gamma\in\defEns{10^{-3}, 10^{-2}, 10^{-1}}$ and a number of iterations equal to $10^6$. ULA diverges for all step sizes. We plot the boxplots of the errors for TULAc, for the first and second moment of the first and last coordinate in \Cref{figure:bad-gaus-tulac}. Although the results for the first coordinate are expectedly inaccurate, the results for the last coordinate are valid. In this context, TULAc enables to obtain relevant results for the well-conditioned coordinates within a relatively small number of iterations, which is not possible using ULA.

\begin{figure}
%\begin{center}
\hspace*{-2cm}\includegraphics[scale=0.5]{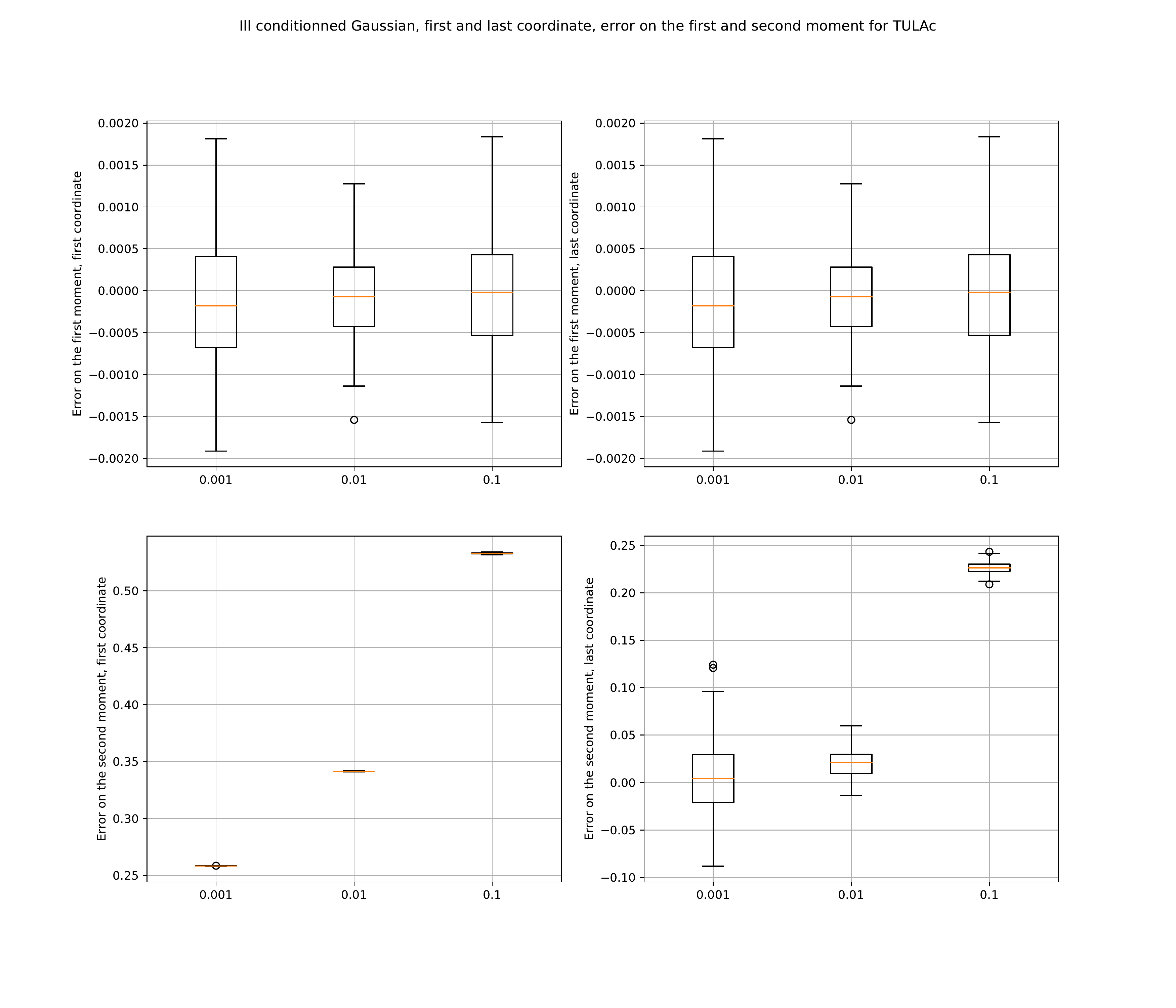}
%\end{center}
\caption{\label{figure:bad-gaus-tulac} Boxplots of the error for TULAc on the first and second moments for the badly conditioned Gaussian variable in dimension $100$ starting at $0$ for different step sizes.}
\end{figure}